\documentclass[twocolumn, aps, pra, superscriptaddress, floatfix, nofootinbib, longbibliography, 10pt]{revtex4-2}
\usepackage{amsmath,bbm,bm}
\usepackage{amssymb}
\usepackage{mathtools}
\usepackage{lipsum}
\usepackage{graphicx}
\usepackage{color}
\usepackage{comment}
\usepackage{amsthm}
\usepackage{ dsfont }
\usepackage{float}
\usepackage{cancel}
\usepackage[autostyle]{csquotes}

\usepackage[colorinlistoftodos]{todonotes}
\usepackage[colorlinks=true, allcolors=blue]{hyperref}
\usepackage{enumitem}

\newtheorem{theorem}{Theorem}
\newtheorem{definition}{Definition}
\newtheorem{lemma}{Lemma}

\newtheorem{proposition}{Proposition}

\newtheorem*{theorem*}{Theorem}

\newcommand{\ket}[1]{\ensuremath{\left|#1\right\rangle}}
\newcommand{\bra}[1]{\ensuremath{\left\langle#1\right|}}

\newcommand{\ketbra}[2]{|#1\rangle\langle#2|}

\def\opone{\leavevmode\hbox{\small1\normalsize\kern-.33em1}}

\newcommand{\dd}{\mathrm{d}}
\newcommand{\tr}{\mathrm{tr}}
\newcommand{\Var}{\mathrm{Var}}

\newcommand{\bfx}{\mathbf{x}}
\newcommand{\bfp}{\mathbf{p}}
\newcommand{\bfa}{\mathbf{a}}
\newcommand{\bfb}{\mathbf{b}}
\newcommand{\bfR}{\mathbf{R}}
\newcommand{\bfr}{\mathbf{r}}
\newcommand{\bfd}{\mathbf{d}}
\newcommand{\bfc}{\mathbf{c}}
\newcommand{\bfs}{\mathbf{s}}
\newcommand{\bfu}{\mathbf{u}}
\newcommand{\bfv}{\mathbf{v}}

\newcommand{\bff}{\mathbf{f}}
\newcommand{\bfg}{\mathbf{g}}

\newcommand{\bfgamma}{\boldsymbol{\gamma}}
\newcommand{\bfnu}{\boldsymbol{\nu}}

\newcommand{\rhoth}{\rho_{\text{th}}}

\newcommand{\E}{\mathbb{E}}
\newcommand{\cG}{\mathcal{G}}
\newcommand{\cQ}{\mathcal{Q}}
\newcommand{\cE}{\mathcal{E}}
\newcommand{\cB}{\mathcal{B}}
\newcommand{\cL}{\mathcal{L}}
\newcommand{\cK}{\mathcal{K}}
\newcommand{\fS}{\mathfrak{S}}

\begin{document}

\title{Quantum Maximal Correlation for Gaussian States}
\date{\today}
\author{Salman Beigi}
\affiliation{School of Mathematics, Institute for Research in Fundamental Sciences (IPM), P.O. Box 19395-5746, Tehran, Iran}

\author{Saleh Rahimi-Keshari}
\affiliation{School of Physics, Institute for Research in Fundamental Sciences (IPM), P.O. Box 19395-5531, Tehran, Iran}
\affiliation{Department of Physics, University of Tehran, P.O. Box 14395-547, Tehran, Iran}

\begin{abstract}
	We compute the quantum maximal correlation for bipartite Gaussian states of continuous-variable systems. Quantum maximal correlation is a measure of correlation with the monotonicity and tensorization properties that can be used to study whether an arbitrary number of copies of a resource state can be locally transformed into a target state without classical communication, known as the local state transformation problem. We show that the required optimization for computing the quantum maximal correlation of Gaussian states can be restricted to local operators that are linear in terms of phase-space quadrature operators. This allows us to derive a closed-form expression for the quantum maximal correlation in terms of the covariance matrix of Gaussian states. Moreover, we define Gaussian maximal correlation based on considering the class of local hermitian operators that are linear in terms of phase-space quadrature operators associated with local homodyne measurements. This measure satisfies the tensorization property and can be used for the Gaussian version of the local state transformation problem when both resource and target states are Gaussian. We also generalize these measures to the multipartite case. Specifically, we define the quantum maximal correlation ribbon and then characterize it for multipartite Gaussian states.
\end{abstract}
\maketitle


\section{Introduction}\label{sec:introduction}

The problem of preparing a desired bipartite quantum state from some available resource states under certain operations is of great foundational and practical interest in quantum information science. This problem has been extensively studied under the class of local operations and classical communication in the context of entanglement distillation, where two parties aim to prepare a highly-entangled state using copies of a weakly entangled state~\cite{BennettPRA1996,BennettPRL1996,BennettPRA1996-2,HorodeckiRMP2009}. The optimal \emph{rate} of entanglement distillation for pure states equals the entanglement entropy~\cite{BennettPRA1996}, and there are mixed entangled states that are not distillable~\cite{HorodeckiPRL98}. 

A more recent version of this problem is \emph{local state transformation} under local operations and \emph{without} classical communication~\cite{Beigi13}. See Figure~\ref{fig:setup} for a precise description of the problem. This problem is highly non-trivial even if the goal is to generate only a single copy of the \emph{target state} $\sigma_{A'B'}$ 
using arbitrarily many copies of the \emph{resource state} $\rho_{AB}$. The difficulty remains even if the target state is not entangled, or even in the fully classical setting~\cite{GhaziFOCS2016}; see also~\cite{QinYao2021}.

To study the local state transformation problem we need measures of correlation that are monotone under local operations and remain unchanged when computed on multiple copies of a bipartite state. The latter crucial property is called \emph{tensorization}, and is required since in local state transformation we assume the availability of arbitrary many copies of the resource state while we aim to generate only a single copy of the target. There are resource measures, based on certain free operations, for Gaussian states~\cite{LamiPRA2018}, quantum coherence~\cite{LamiPRL2019}, and  nonclassicality of quantum states~\cite{LeePRR2022} that satisfy the tensorization property. However, this property is not satisfied by most measures of correlation, such as mutual information and entanglement measures, which makes them inapplicable to the local state transformation problem. \emph{Quantum maximal correlation} was introduced as a measure that satisfies both the monotonicity and the tensorization properties, and therefore is suitable for proving bounds on this problem~\cite{Beigi13}. In particular, based on these properties, one can see that local state transformation is not possible if the maximal correlation of the target state is larger than that of the resource state; see Section~\ref{sec:MC} for more details.
	
Quantum maximal correlation of a bipartite quantum state $\rho_{AB}$ is defined as $\sup |\tr(\rho_{AB} X_A^\dagger\otimes Y_B)|$, where the supremum is taken over all local operators $X_A, Y_B$ that satisfy $\tr(\rho_A X_A) = \rho(\rho_BY_B) =0$ and $\tr(\rho_A X_A^\dagger X_A) = \rho(\rho_BY_B^\dagger Y_B) =1$. It is shown that computing quantum maximal correlation is a tractable problem for systems with a finite-dimensional Hilbert space~\cite{Beigi13}. However, in general, it is not clear how to calculate quantum maximal correlation for states of continuous-variable systems with infinite-dimensional Hilbert spaces, as the space of local operators $X_A, Y_B$ becomes intractable. Of particular interest is the class of Gaussian states that are readily available in the laboratory and can be used as resource states to prepare other states.

In this paper, we compute the quantum maximal correlation for Gaussian states of continuous-variable systems. This measure enables us to study the local state transform problem when either the resource state or the target state is Gaussian. In this case, we show that for Gaussian states it is sufficient to optimize over local operators $X_A, Y_B$ that are linear in terms of phase-space quadrature operators. This turns the inherently infinite-dimensional problem of computing quantum maximal correlation into a finite dimensional one; see Theorem~\ref{thm:main-mc} for the statement of our main result. Moreover, we define \emph{Gaussian maximal correlation}, as a measure based on hermitian and linear local operators in terms of quadrature operators, corresponding to local homodyne measurements. This measure can be used in the Gaussian local state transform scenarios, where the target and resource states are Gaussian. In particular, this shows that copies of weakly-correlated Gaussian states cannot be locally transformed to a highly-correlated Gaussian state. This result should be compared to previous results showing that entanglement in Gaussian states cannot be distilled using Gaussian operations~\cite{EiserScheelPlenio,GiedkeCirac02,FiurasekPRL2002,GiedkeKraus14}.

We also generalize the quantum maximal correlation to the multipartite setting. We define an invariant of multipartite quantum states called \emph{quantum maximal correlation ribbon} that, similar to the bipartite case, satisfies the monotonicity and the tensorization properties. We also show that to compute the maximal correlation ribbon for multipartite Gaussian states it suffices to restrict to local operators that are linear in terms of quadrature operators.

\begin{figure}
	\begin{center}
		\includegraphics[scale=0.53]{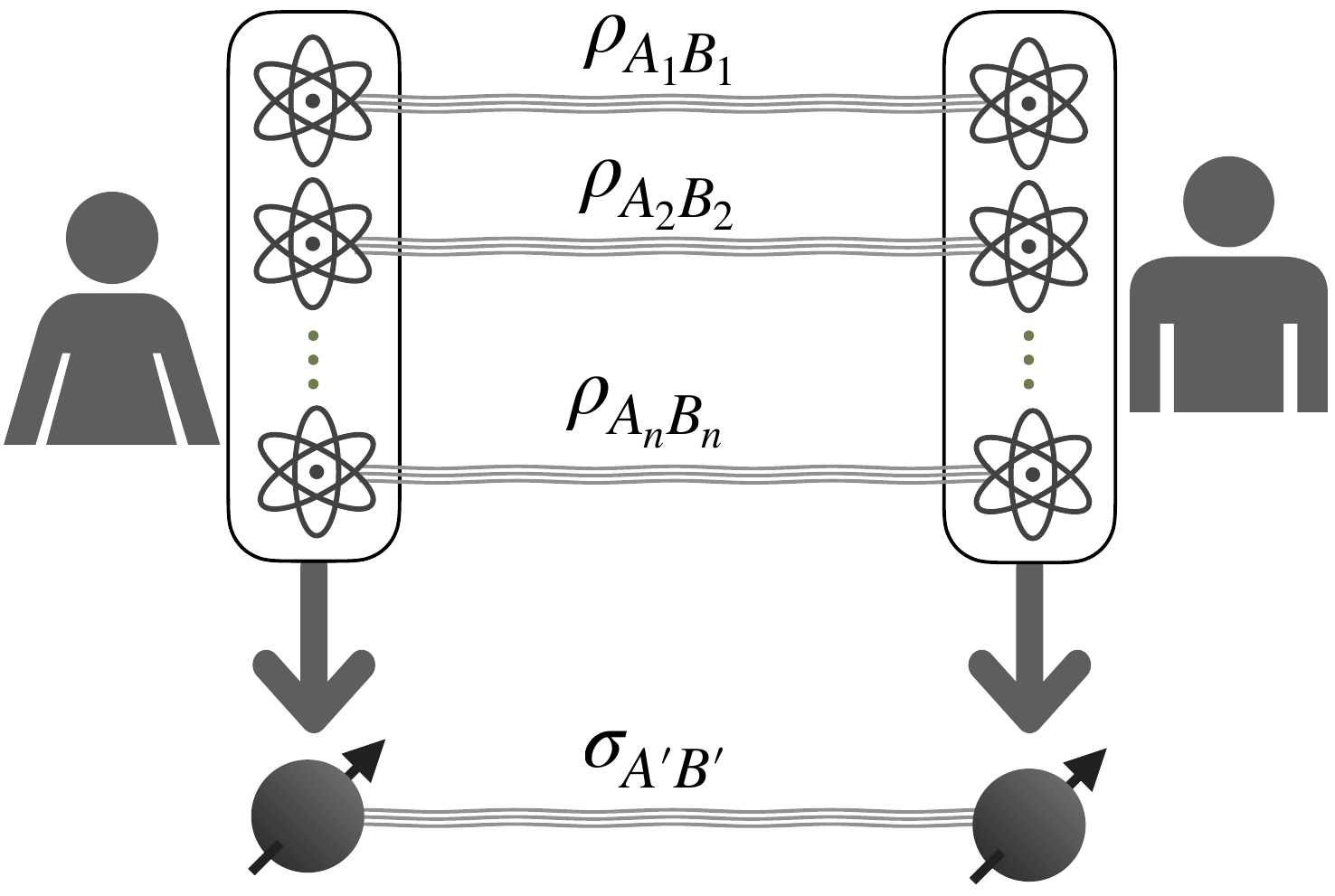}
		\caption{\footnotesize  The local state transformation problem: two parties share an arbitrary number of copies of a resource state $\rho_{AB}$, and their goal is to generate a single copy of a target state $\sigma_{A'B'}$. They are allowed to apply local quantum operations, but classical communication is forbidden. Quantum maximal correlation, a useful measure to study the local state transformation problem, is monotone under local operations and remains unchanged when computed on multiple copies of the state. The latter is known as the tensorization property. The local state transformation is not possible if the quantum maximal correlation of $\sigma_{A'B'}$ is larger than that of $\rho_{AB}$. 
		}
		\label{fig:setup}
	\end{center}
\end{figure}  

\medskip

\noindent
\emph{Structure of the paper:}  In Section~\ref{sec:MC}, we review the notion of quantum maximal correlation from~\cite{Beigi13}. In Section~\ref{sec:GQS}, we review the definition of Gaussian states and some of their main properties. We note that the objective function in quantum maximal correlation is a bilinear form which can be thought of as an inner product. Then, to compute maximal correlation it would be useful to find orthonormal bases for the space of local operators $X_A, Y_B$. Following this point of view, Section~\ref{sec:orth-basis} is devoted to introduce such an orthonormal basis for Gaussian states which might be of independent interest. Results of this section will be used to prove our main result in Section~\ref{sec:bipart-Gaussian}. We introduce the Gaussian maximal correlation in Section~\ref{sec:GMC}. Considering examples of Gaussian states in Section~\ref{sec:Examples}, we further illustrate various features of these correlation measures. We also generalize the definition of maximal correlation for multipartite states in Section~\ref{sec:MC-ribbon-Gaussian} by introducing the quantum maximal correlation ribbon. We show that maximal correlation ribbon satisfies the monotonicity and the tensorization properties, and compute it for multipartite Gaussian states. Detailed proofs of the results in the multipartite case are left for the appendices. We conclude the paper in Section~\ref{sec:discussion}.

\section{Review of Quantum Maximal Correlation}\label{sec:MC}

Given a bipartite probability distribution $p_{AB}$, its maximal correlation is defined as the maximum of the \emph{Pearson correlation coefficient} over all functions of random variables $A$ and $B$~\cite{Hirschfeld1935,Gebelein1941,Renyi1959Measures,Witsenhausen1975}. That is, the classical maximal correlation is given by\footnote{In the classical literature, maximal correlation is usually denoted by $\rho$. Here in the quantum case, following~\cite{Beigi13}, we save $\rho$ for density matrices and denote maximal correlation by $\mu$.}
\begin{align}\label{eq:MC-classic}
\mu(A, B) = \sup \E[f_A g_B],
\end{align}
where the supremum is taken over all real functions $f_A$ and $g_B$ with zero mean and unit variance, i.e., $\E[f_A]=\E[g_B] =0$ and $\E[f_A^2] = \E[g_B^2]=1$. The classical maximal correlation is zero if and only if the two random variables are independent, and equals one if they have a ``common bit," meaning that there are non-trivial functions $f, g$ such that $f(A)= g(B)$.

Maximal correlation for Gaussian distributions is first computed in~\cite{Lancaster}. The main result of~\cite{Lancaster} is that in the computation of maximal correlation for a Gaussian distribution $p_{AB}$ over $\mathbb R^2$, it suffices to restrict the optimization in~\eqref{eq:MC-classic} to functions $f_A$ and $g_B$ that are linear in $A$ and $B$, respectively. Based on this, there are essentially unique choices of linear functions $f_A$ and $g_B$ and then $\mu(A, B)$ for such a Gaussian distribution $p_{AB}$ equals
\begin{equation}\label{eq:ClasMC-Gaus}
\mu(A, B) =\frac{|\text{Cov}[A,B]|}{\sqrt{\text{Var}[A]\text{Var}[B]}},
\end{equation}
where $\text{Cov}[A,B]=\E[A B]-\E[A]\E[B]$ and $\text{Var}[A]=\text{Cov}[A,A]$ is the variance.

Maximal correlation in the quantum case can be defined by replacing functions in Eq.~(\ref{eq:MC-classic}) with  local operators~\cite{Beigi13}. Given a density operator $\rho_{AB}$ describing the joint state of a bipartite quantum system, quantum maximal correlation is defined as
\begin{align}\label{eq:def-max-cor}
\mu(A, B) := \sup_{X_A, Y_B} & |\tr\big(\rho_{AB} X_A^\dagger\otimes Y_B\big)|\\
& ~\tr(\rho_A X_A) = \rho(\rho_BY_B) =0,\nonumber\\
& ~\tr(\rho_A X_A^\dagger X_A) = \rho(\rho_BY_B^\dagger Y_B) =1,\nonumber
\end{align}
where the supremum is over all choices of local operators $X_A, Y_B$. Note that, in this definition, these operators are not necessarily hermitian, hence we maximize over the modulus of the objective function $\tr(\rho_{AB} X_A^\dagger\otimes Y_B)$. In fact, as shown in~\cite{Beigi13}\footnote{See the updated arXiv version of~\cite{Beigi13}.} and discussed in Section~\ref{sec:Examples}, for some quantum states non-hermitian operators are required to further optimize the quantum maximal correlation.

An application of the Cauchy-Schwarz inequality shows that $0\leq \mu(A, B)\leq 1$. It can be shown that $\mu(A, B)=0$ if and only if $\rho_{AB}=\rho_A\otimes \rho_B$ is a product state. Note that if $\rho_{AB}$ is not product, there always exist local measurements whose outcomes are correlated and hence their maximal correlation is nonzero. Also, restricting $X_A$ and $Y_B$ to local hermitian operators, one can verify that $\mu(A, B)=1$ if and only if there exist non-trivial local measurements described by operators $\{\Pi_{A,n}\}_{n=1}^{d}$ and $\{\Pi_{B,m}\}_{m=1}^{d}$, where $d\geq2$, such that the joint outcome probability distribution is perfectly correlated~\cite{Beigi13}: 
\begin{equation}\label{eq:perfect-shar-rand}
p(n,m)=\tr\big(\rho_{AB} \Pi_{A,n}\otimes \Pi_{B,m}\big)=p_n\delta_{n,m}.
\end{equation}
This form means that the shared state between two parties can be used to extract perfect shared randomness through local measurements. 
Based on this, we can see that $\mu(A, B)=1$ for any pure entangled state $\ket{\psi_{AB}}=\sum_n c_n\ket{n_A,n_B}$ with $\{\ket{n_A}\}$ and $\{\ket{n_B}\}$ being the Schmidt bases, and its completely dephased version in the Schmidt bases, $\rho_{AB}=\sum_n |c_n|^2 \ketbra{n_A}{n_A}\otimes\ketbra{n_B}{n_B}$, which is a separable state. In these cases, considering observables $N_A=\sum_n h_n\ketbra{n_A}{n_A}$ and $N_B=\sum_n h_n\ketbra{n_B}{n_B}$ for subsystems $A$ and $B$, the optimal local operators in~\eqref{eq:def-max-cor} can be found as  $X_A=(N_A-\langle N_A\rangle)/\Delta N_A$ and $Y_B=(N_B-\langle N_B\rangle)/\Delta N_B$ with the identical means $\langle N_A\rangle=\langle N_B\rangle=\sum_n |c_n|^2 h_n$, and the identical standard deviations, $\Delta N_A=\Delta N_B=(\sum_n |c_n|^2 h_n^2-\langle N_A\rangle^2)^{1/2}$.

Maximal correlation satisfies two crucial properties
~\cite{Beigi13}:
\begin{itemize}
\item \textbf{(Data processing)} If $\sigma_{A'B'} = \Phi_{A\to A'}\otimes \Psi_{B\to B'}(\rho_{AB})$, where $\Phi_{A\to A'}$ and $\Psi_{B\to B'}$ are local quantum operations (quantum completely-positive trace preserving (cptp) super-operators) then 
$$\mu(A', B') \leq \mu(A, B).$$

\item \textbf{(Tensorization)} For any bipartite density matrix $\rho_{AB}$ and any integer $n$ we have
$$\mu(A^n, B^n) = \mu(A, B),$$
where the left hand side is the maximal correlation of the state $\rho_{A^nB^n}=\rho_{AB}^{\otimes n}$ considered as a bipartite state. 
\end{itemize}

The first property says that maximal correlation is really a measure of correlation, and does not increase under local operations. The second property, however, is an intriguing property saying that no matter how many copies of $\rho_{AB}$ are shared between two parties, their maximal correlation remains the same. This is unlike most measures of correlations (such as mutual information and entanglement entropy) that scale when the number of copies increase. 

The above two properties of maximal correlation make it suitable for proving impossibility of certain local state transformations. Suppose that $\rho_{AB}$ and $\sigma_{A'B'}$ are two bipartite state with $\mu(A, B)<\mu(A', B')$. Then, it is not possible to transform $\rho_{AB}$ to $\sigma_{A'B'}$ under local operations even if an arbitrary many copies of $\rho_{AB}$ is available.
This is because if there exist local operations $\Phi_{A^n\to A'}, \Psi_{B^n\to B'}$ acting on $n$ copies of $\rho_{AB}$ such that $\sigma_{A'B'} = \Phi_{A^n\to A'}\otimes \Psi_{B^n\to B'}(\rho_{AB}^{\otimes n})$, then we must have 
$$\mu(A, B) = \mu(A^n, B^n)\geq \mu(A', B'),$$
where the equality is due to the tensorization property and the inequality follows from the data processing inequality.

Let us examine the example of  noisy Bell states. Let 
\begin{equation}\label{eq:Werner state}
\rho_{AB}^{(\kappa)} = \kappa\ketbra{\psi}{\psi} + (1-\kappa) \frac 14 {\bm I}, 
\end{equation}
be a mixture of the Bell state $\ket\psi = \frac{1}{\sqrt 2} (\ket{00} + \ket{11})$ and the maximally mixed state ${\bm I}/4$ for two qubits with $0\leq\kappa\leq1$.  It is not hard to verify that $\mu(A, B) = \kappa$; see~\cite{Beigi13}. Thus, by the above observation, local transformation of $\rho_{AB}^{(\kappa)}$ to $\rho_{AB}^{(\kappa')}$ is impossible if $\kappa'>\kappa$ even with arbitrary many copies of $\rho_{AB}^{(\kappa)}$. 

It is shown in~\cite[Theorem 1]{Beigi13} that the maximal correlation $\mu(A, B)$ of $\rho_{AB}$ is equal to the \emph{second Schmidt coefficient} of some vector associated to $\rho_{AB}$ in a tensor product Hilbert space. This makes the problem of computing $\mu(A, B)$ tractable when the dimensions of the local Hilbert spaces of $A, B$ are finite. In the infinite dimensional case, however, computation of maximal correlation remains a challenge in general.  

Our main result in this paper is to compute the maximal correlation for Gaussian quantum states, which are of great practical interest in quantum information processing. Similar to the classical maximal correlation for Gaussian probability distributions, we show that in the computation of $\mu(A, B)$ for Gaussian state $\rho_{AB}$, the optimization in~\eqref{eq:def-max-cor} can be restricted to operators $X_A, Y_B$ that are linear in local creation and annihilation operators.

\section{Gaussian Quantum States}\label{sec:GQS}

In this section we review the  definition and some basic properties of Gaussian states. For more details we refer to~\cite{Serafini17, Adesso+14}.

The Hilbert space of an $m$-mode bosonic quantum system  is the space of square-integrable functions on $\mathbb R^{2m}$. We denote the annihilation and creation operators of the $j$-th mode by $\bfa_j = (\bfx_j+i\bfp_j)/\sqrt 2$ and $\bfa_j^\dagger = (\bfx_j-i\bfp_j)/\sqrt 2$, where $ \bfx_j$ and $\bfp_j$ are the phase-space quadrature operators, similar to the position and momentum operators of the quantum harmonic oscillator, that satisfy the commutation relations\footnote{We assume that $\hbar=1$.}
$$[\bfx_j, \bfp_k] = i\delta_{jk}.$$
We denote by 
$$\mathbf{ R}=(\bfx_1, \bfp_1, \dots, \bfx_m, \bfp_m)^\top,$$
the column-vector consisting of quadrature operators. Then, the above commutation relations can be summarized as 
\begin{align}\label{eq:CCR}
[\bfR, \bfR^\top] = i\Omega_m,
\end{align}
where $[\bfR, \bfR^\top]$ is understood as coordinate-wise commutataion and $\Omega_m= \Omega_1\oplus\cdots \oplus \Omega_1$ with
$$\Omega_1 = \begin{pmatrix}
0 & 1 \\
-1 & 0
\end{pmatrix}.
$$
For $\mathbf r =(r_1, \dots, r_{2m})^\top\in \mathbb R^{2m}$, the $m$-mode \emph{displacement operator} (also called the \emph{Weyl operaator}) $D_{\mathbf r}$ is defined by
\begin{align}\label{eq:disp-op}
	\begin{split}
D_{\bfr} &= e^{-i\bfr^\top \Omega_m \bfR}
= e^{i\sum_{j=1}^m (r_{2j} \bfx_{j}-r_{2j-1}\bfp_j )}\\
& = \bigotimes_{j=1}^m e^{(r_{2j} \bfx_{j}-r_{2j-1}\bfp_j)}
= \bigotimes_{j=1}^m D_{(r_{2j-1}, r_{2j})},
	\end{split}
\end{align}
which is the tensor product of the displacement operators for each mode, and hence is local. We note that $D_{\bfr}^{\dagger} = D_{-\bfr}$. Moreover, as a consequence of the Baker-Campbell-Hausdorff (BCH) formula\footnote{If $[X, [X, Y]]=[Y, [X, Y]]=0$, then $e^{X+Y} = e^Xe^Y e^{-\frac 12[X, Y]}$.} we have
$$D_{\bfr}D_{\bfs} = e^{-i \frac 12\bfr^\top\Omega_m \bfs} D_{\bfr+\bfs}.$$
This, in particular, means that $D_{\bfr}D_{\bfr}^\dagger=D_{\bfr}D_{-\bfr} = D_{\mathbf{0}}= I$, i.e., the displacement operator is unitary. 

A crucial property of $D_{\bfr}$ is that\footnote{For the proof note that if $[X, [X, Y]] = 0$, then $e^{X}Ye^{-X}= Y+[X, Y]$. Also, see~\cite[Equation (3.16)]{Serafini17}.} 
\begin{align}\label{eq:disp-equation}
D_{\bfr} \bfR D_{\bfr}^\dagger= \bfR-\bfr.
\end{align}
This means that $D_{\bfr}\bfx_j D_{\bfr}^\dagger = \bfx_j - r_{2j-1}$ and $D_{\bfr} \bfp_j D_{\bfr}^\dagger = \bfp_j - r_{2j}$, i.e., $D_{\bfr}$ displaces quadrature operators under conjugation, thus the name.

For an arbitrary $m$-mode density operator $\rho = \rho_{A_1, \dots, A_m}$, the vector of \emph{first-order moments} in $ \mathbb R^{2m}$ is defined as 
\begin{align}\label{eq:def-r-rho}
\bfd(\rho) = \tr(\rho \bfR),
\end{align}
which contains the canonical mean values $\bfd(\rho)_{2j-1} = \tr(\rho \bfx_j)$ and $\bfd(\rho)_{2j} = \tr(\rho\bfp_j)$. The \emph{covariance matrix}, containing the \emph{second-order moments}, is defined by
\begin{equation}\label{eq:CovMat-def}
\bfgamma(\rho) = \tr[\rho\{(\bfR- \bfd(\rho)), (\bfR- \bfd(\rho))^\top\}],
\end{equation}
where $\{\cdot, \cdot\}$ denotes anti-commutation, and as before $\{(\bfR- \bfd(\rho)), (\bfR- \bfd(\rho))^\top\}$ is understood as coordinate-wise anti-commutation. Thus, $\bfgamma(\rho)$ is a $(2m)\times (2m)$ matrix, that by definition is real and symmetric. Furthermore, it can be shown (see~\cite[Equation (3.77)]{Serafini17}) that as a consequence of the canonical commutation relations~\eqref{eq:CCR} we have
\begin{align}\label{eq:uncertainty}
\bfgamma(\rho)+i\Omega_m\geq 0.
\end{align}
This, in particular, means that $\bfgamma(\rho)$ is positive definite.\footnote{Taking the transpose of~\eqref{eq:uncertainty} we find that  $\bfgamma(\rho)-i\Omega_m\geq 0$. Summing this with the original equation gives $\bfgamma(\rho)\geq 0$. To verify that $\bfgamma(\rho)$ does not have a zero eigenvalue, suppose that $\bfgamma(\rho)\bfr=0$ and for $\bfs=\bfr - \epsilon i \Omega_m \bfr$ write down the condition $\bfs^\dagger \bfgamma(\rho)\bfs \geq -i\bfs^\dagger \Omega_m \bfs$ to conclude that $\bfr=0$.} 

Let $\bfgamma_j(\rho)$ be the $j$-th $2\times 2$ block on the diagonal of $\bfgamma(\rho)$. Then, by definition, the covariance matrix of $\rho_{A_j}$, the marginal state on the $j$-th mode, equals $\bfgamma(\rho_{A_j})=\bfgamma_j(\rho)$. Similarly, $\bfd(\rho_{A_j})$, the first-order moments of the marginal state, equals the $j$-th pair of components of $\bfd(\rho)$.

Let us examine the effect of the displacement operator on the first and second moments of a density operator. Let $\rho'=D_{\bfr}^\dagger\rho D_{\bfr}$. Then, by~\eqref{eq:disp-equation} we have
\begin{align}\label{eq:disp-shift-FM}
\bfd(\rho')  = \tr(\rho'\bfR)=\tr(\rho D_{\bfr}\bfR D_{\bfr}^\dagger)= \bfd(\rho)- \bfr.
\end{align}
This implies that  $D_{\bfr}(\bfR-\bfd(\rho')) D_{\bfr}^\dagger = \bfR - \bfd(\rho)$ and therefore by \eqref{eq:CovMat-def} we get $\bfgamma(\rho')=\bfgamma(\rho)$. Thus, the application of the local unitary $D_{\bfr}$ shifts the first-order moment of $\rho$ but does not change the covariance matrix.

The \emph{characteristic function} of a density operator $\rho$ is defined by 
\begin{align*}
\chi(\bfr) = \tr(\rho D_{\bfr}). 
\end{align*}
Characteristic function fully determines a density operator via 
$$\rho = \frac{1}{(2\pi)^m} \int_{\mathbb R^{2m}} \!\!\chi(\bfr) D_{-\bfr}\, \dd^{2m} \bfr.$$
The \emph{Wigner function} function~\cite{Wigner} is then defined as the Fourier transform of the characteristic function
$$W(\bfs) = \frac{1}{(2\pi)^{2m}}\int_{\mathbb R^{2m}}\!\! \chi(\bfr) e^{i\bfr^\top \Omega_m \bfs}\dd^{2m} \bfr.$$
Applying the inverse Fourier transform we obtain
$$\tr\big(\rho e^{-i\bfr^\top \Omega_m \bfR}\big) =\chi(\bfr) = \int_{\mathbb R^{2m}}\!\! W(\bfs) e^{-i\bfr^\top \Omega_m \bfs}\dd^{2m} \bfs.$$
We note that this equation holds for any $\bfr\in \mathbb R^{2m}$. Thus, thinking of both sides as functions of $\bfr$, considering their Taylor expansion and comparing corresponding terms, we realize that the same equation holds for all $\bfr$ beyond real ones. That is, for any
complex $\bfc\in \mathbb C^{2m}$, we have 
\begin{align}\label{eq:tr-Wigner}
\tr(\rho e^{\bfc^\top  \bfR}) = \int_{\mathbb R^{2m}}\!\! W(\bfs) e^{\bfc^\top  \bfs}\dd^{2m} \bfs.
\end{align}

A quantum state is called \emph{Gaussian} if its Wigner function is Gaussian. In this case, the Wigner function is specified by the first-order moments $\bfd=\bfd(\rho)$ and the covariance matrix $\bfgamma=\bfgamma(\rho)$:
\begin{align}\label{eq:gaussian-W}
W(\bfs) = \frac{1}{\pi^m \sqrt{\det \bfgamma}} e^{-(\bfs-\bfd)^\top \bfgamma^{-1}(\bfs-\bfd)}.
\end{align}
By using the Wigner function, one can verify that the purity of Gaussian states can be calculated as $\tr(\rho^2)=(2\pi)^m\int_{\mathbb R^{2m}} W^2(\bfs) \dd^{2m}\bfs=1/\sqrt{\det \gamma}$. Therefore, pure Gaussian states have $\det \gamma=1$.
 
The Wigner function of a marginal state is the marginal distribution of the Wigner function of the joint state. Hence, marginal states of Gaussian state are also Gaussian. Note also that, due to the uncertainty principle, the Wigner function cannot be viewed as the joint probability distribution associated with local measurements.

Coherent states are well-known examples of a single-mode Gaussian states, which are displaced vacuum states. For a complex number $\alpha=(x+ip)/\sqrt{2}$ let $\ket{\alpha}=D_{\alpha}\ket{0}$ where $D_{\alpha}=\exp(\alpha \bfa^\dagger-\bar\alpha\bfa)=D_{(x, p)}$. Then,  $\bfd(\ket{\alpha})=(x,p)^\top$ and $\bfgamma(\ket \alpha)=I$, where $I$ is the $2\times2$ identity matrix\footnote{In this paper, $I_n$ denotes the $n\times n$ identity matrix but we drop the subindex for $n=2$.}. Other important examples are squeezed-vacuum states $\ket{z}=\exp(z(\bfa^{\dagger2}-\bfa^2)/2)\ket{0}$ with squeezing parameter $z$ and thermal states $\rho_{\text{th}}=\int  e^{-|\alpha|^2/\bar{n}}/(\bar{n}\pi)\ketbra{\alpha}{\alpha} d^2\alpha$ with mean-photon number $\bar{n}$. The first-order moments of these states are zero and their covariance matrices are $\bfgamma(\ket{z})=\text{diag}(e^{2z},e^{-2z})$ and $\bfgamma(\rho_{\text{th}})=(2\bar{n}+1)I$.

In general, Gaussian unitary operators that map Gaussian states to Gaussian ones can be described in terms of a displacement operator and a unitary operator associated to Hamiltonians that are quadratic in terms of quadrature operators~\cite[Chapter 5]{Serafini17}. The latter 
can be written in the form of $U_H=e^{-\frac 12 i\bfR^\top H\bfR}$ in which $H$ is a $(2m)\times (2m)$ real symmetric matrix. Letting $\rho' = U_H \rho U_H^\dagger$, it can be shown that  
\begin{align}\label{eq:S_H-r}
\bfd(\rho') = S_H \bfd(\rho),
\end{align}
and 
\begin{align}\label{eq:S_H-gamma}
\bfgamma(\rho') = S_H \bfgamma(\rho)S_H^\top,
\end{align}
where $S_H = e^{\Omega_m H}$ is a \emph{symplectic matrix} (satisfying $S_H\Omega_m S_H^\top = \Omega_m$). 
Examples of single-mode unitary operators are the phase rotation with $H_\text{Ro}=\theta(I-Y)$ and the symplectic matrix $S_{\text{Ro}}=\cos(\theta) I+i\sin(\theta) Y$, and squeezing  with $H_\text{Sq}=zX$ and the symplectic matrix $S_{\text{Sq}}=\text{diag}(e^z,e^{-z})$, where here $Y$ and $X$ are the Pauli matrices. Using these two unitary operations and displacement, any single-mode pure Gaussian state can be transformed to the vacuum state.

The following proposition provides a \emph{standard form} for Gaussian states under local unitaries.

\begin{proposition}\label{prop:standard-form}
Let $\rho=\rho_{A_1, \dots, A_m}$ be an $m$-mode Gaussian state. Then, there exists a \emph{local} Gaussian unitary $V$ such that for $\rho'=V\rho V^\dagger$ we have $\bfd(\rho') = 0$ and  $\bfgamma_j(\rho') = \bfgamma(\rho'_{A_j}) =\lambda_j I$ with $\lambda_j\geq 1$. If $m=2$ we can further assume that the covariance matrix can be in the standard form
\begin{equation}\label{eq:CM-StandForm}
	\bfgamma(\rho') = \begin{pmatrix}
\lambda_1 I & \bfnu\\
\bfnu^\top & \lambda_2 I
\end{pmatrix},
\end{equation}
with $\bfnu=\text{diag}(\nu_1,\nu_2)$. 
\end{proposition}

\begin{proof}
	First, by applying an appropriate displacement operator, that is a local unitary, we can shift the first-order moments to zero without changing the covariance matrix. Next, to bring the state into the desired form, we use a local unitary operator $U_H=e^{-\frac 12 i\bfR^\top H\bfR}$, where $H$ is block-diagonal with $2\times 2$ blocks on the diagonal (one block for each mode). This means that block-diagonal symplectic matrices correspond to local unitaries. 

Now, suppose that $S_H = \text{diag}(S_1, \dots, S_m)$ is a block-diagonal symplectic matrix with $S_j$'s being $2\times 2$ symplectic matrices to be determined. Also, let $\bfgamma_j(\rho) = \bfgamma(\rho_{A_j})$ be the $j$-th $2\times 2$ block on the diagonal of the covariance matrix. By~\eqref{eq:S_H-gamma} we know that the application of the local unitary $U_H$ would change $\bfgamma_j(\rho)$ to $S_j \bfgamma_j(\rho) S_j^\top$. By choosing each local unitary to be a phase rotation that diagonalizes $\bfgamma_j(\rho)$, followed by a squeezing operator that makes the diagonal elements equal, we get $\bfgamma'_j=S_{\text{Sq}j}S_{\text{Ro}j}\bfgamma_j S_{\text{Ro}j}^\top S_{\text{Sq}j}^\top =\lambda_j I$ with $\lambda_j\geq 1$. Note that $\lambda_j=1$ means that the marginal state $\rho'_{A_j}$ is the vacuum state, which is pure and hence cannot be correlated with the other modes. Also,
$\lambda_j>1$ implies a thermal marginal state with mean-photon number $\bar{n}=(\lambda_j-1)/2$.
Putting these together we find that there is a local Gaussian unitary operator $V= U_H D_\bfr$, consisting of local displacements, phase rotations and squeezing, that brings the covariance matrix to the desired form with marginal thermal states.

For $m=2$, the above procedure gives a covariance matrix of the form 
$$\begin{pmatrix}
\lambda_1 I & \bfnu'\\
\bfnu'^\top & \lambda_2 I
\end{pmatrix}.$$
However, by including further local phase-rotation operations with the block-diagonal symplectic matrix $S_{H'}=\text{diag}(S'_{\text{Ro}1},S'_{\text{Ro}2})$ into $V$, we can get 
$$\begin{pmatrix}
	\lambda_1 I & S'_{\text{Ro}1} \bfnu' {S'_{\text{Ro}2}}\!^\top\\
	S'_{\text{Ro}2}\bfnu'^\top {S'_{\text{Ro}1}}\!^\top & \lambda_2 I
\end{pmatrix},$$
where $S'_{\text{Ro}1} \bfnu' {S'_{\text{Ro}2}}\!^\top=\bfnu$ is diagonal. Thus, we obtain the standard form of the covariance matrix for bipartite Gaussian states.
\end{proof}

Another useful phase-space quasiprobability distribution is the Glauber-Sudarshan $P$-function~\cite{Glauber1963,Sudarshan1963} that, in terms of the characteristic function, is given by
\begin{equation}\label{eq:GSP-function}
	P(\bfs) = \frac{1}{(2\pi)^{2m}}\int_{\mathbb R^{2m}}\!\!\chi(\bfr) e^{\bfr^\top\bfr/4} e^{i\bfr^\top \Omega_m \bfs} \,\dd^{2m}\bfr.
\end{equation}
Using this distribution, a density operator of an $m$-mode system can be expressed in terms of $m$-mode coherent states
\begin{equation}
	\rho=\int_{\mathbb C^{m}}\!\!P(\bm\alpha) \ketbra{\bm\alpha}{\bm\alpha}\, \dd^{2m}\bm\alpha,
\end{equation}
where $\bm\alpha=(\alpha_1,\alpha_2,\dots,\alpha_m)^\top\in \mathbb C^{m}$ and by using $\alpha_j=(s_{2j-1}+is_{2j})/\sqrt{2}$ we have $P(\bm\alpha)=2^m P(\bfs)$. For most quantum states the Glauber-Sudarshan $P$-function either takes negative values or is a highly-singular function, existing as a generalized distribution~\cite{Cahill1965}; these states are known as \emph{nonclassical} states~\cite{Mandel1986}. Quantum states whose $P(\bm\alpha)$ is a probability density distribution are known as \emph{classical} states~\cite{Titulaer1965}. For Gaussian states, if 
\begin{equation}
	\tilde\bfgamma(\rho)=\bfgamma(\rho)- I_{2m}\geq0,
\end{equation} 
where $I_{2m}$ is the $2m\times 2m$ identity matrix, then the Fourier transform~\eqref{eq:GSP-function} exists and the Glauber-Sudarshan $P$-function is a Gaussian distribution,
\begin{align}\label{eq:gaussian-P}
	P(\bfs) = \frac{1}{\pi^m \sqrt{\det \tilde\bfgamma}} e^{-(\bfs-\bfd)^\top \tilde\bfgamma^{-1}(\bfs-\bfd)}.
\end{align}
In this case, the state is classical and separable.

\section{An Orthonormal Basis for Local Operators}\label{sec:orth-basis}

In this section, we derive an orthonormal basis for the space of (local) operators with respect to a Gaussian state. Let $\rho$ be an arbitrary quantum state which for the sake of simplicity, is assumed to be full-rank. Then, for any pair of operators $X, X'$ we define
\begin{align}\label{eq:def-inner-prod}
\langle X, X'\rangle_\rho:= \tr(\rho X^\dagger X').
\end{align}
$\langle \cdot, \cdot \rangle_\rho$ satisfies all the properties of an inner product: it is linear in the second argument and anti-linear in the first argument; also, $\langle X, X\rangle_{\rho}\geq 0$ and equality holds iff\footnote{Note that we assume that $\rho$ is full-rank} $X=0$. Note that the objective function in the maximal correlation~\eqref{eq:def-max-cor} can be written in terms of this inner product:
$\big|\langle X_A, Y_B \rangle_{\rho_{AB}}\big|$.
Thus, to compute the maximization in~\eqref{eq:def-max-cor} it would be helpful to compute an orthonormal basis for the space of local operators $X_A, Y_B$ with respect to the inner products associated to
the marginal states $\rho_A=\tr_B(\rho_{AB})$ and $\rho_B=\tr_A(\rho_{AB})$, respectively.

Let $\rho_{AB}$ be a two-mode Gaussian state. It is clear from the definition that local unitaries do not change the maximal correlation. Therefore, by using Proposition~\ref{prop:standard-form} and without loss of generality, we assume that the covariance matrix of $\rho_{AB}$ is in the standard form~\eqref{eq:CM-StandForm} with $\bfd(\rho_{AB})=0$.

In the following, we consider a single-mode thermal state $\rhoth$ with $\bfd(\rhoth)=0$ and $\bfgamma(\rhoth)=\lambda I$ with $\lambda> 1$. This can be either of the marginal states $\rho_A$ or $\rho_B$ of the Gaussian state $\rho_{AB}$.

For any $\bfr \in \mathbb R^{2}$, let 
\begin{align}\label{eq:def-C-r}
C_\bfr =\exp\!\left(\bfr^\top \Upsilon \bfR -\frac \lambda 4 \bfr^\top \Upsilon \Upsilon^\top \bfr\right)\!,
\end{align}
where $\Upsilon$ is a $2\times 2$ matrix to be determined. Using the BCH formula, we have
$$C_{\bfr}^\dagger C_{\bfs} = \eta\, e^{(\bfr^\top\bar \Upsilon+\bfs^\top \Upsilon)\bfR},$$
where $\bar \Upsilon$ is the entry-wise complex conjugate of $\Upsilon$, and
$$\eta=\exp\!\bigg(\!\!\! -\frac{\lambda}{4}\big(\bfr^\top \bar \Upsilon \bar \Upsilon^\top \bfr + \bfs^\top  \Upsilon  \Upsilon^\top \bfs\big) + \frac i 2 \bfr^\top \bar \Upsilon \Omega_1 \Upsilon^\top \bfs\bigg).$$
Then, by using~\eqref{eq:tr-Wigner} and~\eqref{eq:gaussian-W}, we have
\begin{align}\label{eq:tr-rho-CdC}
\tr\big(\rhoth C_{\bfr}^\dagger C_{\bfs}\big) &= \eta \int_{\mathbb R^2}\! W(\bfu)\,  e^{(\bfr^\top\bar \Upsilon+\bfs^\top \Upsilon)\bfu }\,\dd^2 \bfu\nonumber\\
& =\exp\!\left({\frac \lambda 2\bfr^\top \bar \Upsilon \Upsilon^\top \bfs + \frac i 2 \bfr^\top \bar \Upsilon \Omega_1 \Upsilon^\top \bfs}\right)\!,
\end{align}
where $W(\bfu)=\exp(-\bfu^\top \bfu/\lambda)/(\pi\lambda)$ is the Wigner function of the thermal state.
Now let 
\begin{align}\label{eq:def-M}
\Upsilon= \frac{1}{\sqrt 2}\begin{pmatrix}
1 & -i\\
1 & i
\end{pmatrix},
\end{align}
which is unitary and gives
$$\bar \Upsilon \Omega_1 \Upsilon^\top = \begin{pmatrix}
-i & 0\\
0 & i
\end{pmatrix}. $$
This choice of $\Upsilon $ simplifies \eqref{eq:tr-rho-CdC} as
\begin{align} \label{eq:sigma-CC}
\tr\big(\rhoth C_{\bfr}^\dagger C_{\bfs}\big) 
&=\exp\!\bigg(\!\frac{(\lambda+1)r_1 s_1}{2}+\frac{(\lambda-1)r_2 s_2}{2}\bigg)\\
&= \sum_{k, \ell=0}^\infty \frac{1}{k! \ell!} \bigg(\!\frac{(\lambda+1)r_1 s_1}{2}\!\bigg)^{\! k}\!\bigg(\!\frac{(\lambda-1)r_2 s_2}{2}\!\bigg)^{\!\ell}\!.\nonumber
\end{align}

On the other hand, we define the operators $H_{k, \ell}$ by expanding $C_\bfr$ as a function of $\bfr=(r_1, r_2)^\top$:
\begin{align}\label{eq:def-H-k-ell}
C_\bfr = \sum_{k, \ell=0}^\infty     \frac{1}{\sqrt{k!\ell!}} \zeta_0^k(\lambda)\zeta_1^\ell(\lambda) r_1^k r_2^{\ell}  H_{k, \ell},
\end{align}
with 
\begin{align}\label{eq:def-zeta}
\zeta_0(\lambda) = \sqrt{ \frac{(\lambda+1) }{2}}, \quad \zeta_1(\lambda) = \sqrt{ \frac{(\lambda-1) }{2}}.
\end{align}
We note that $H_{k, \ell}$ is a polynomial of the 
quadrature operators $\bfx$ and $\bfp$ of degree $k+\ell$. For instance, $H_{0,0}=I$ and 
\begin{align}\label{eq:H-t=1}
H_{1, 0} = \frac{1}{\sqrt{\lambda+1}}(\bfx-i\bfp), ~ H_{0, 1} = \frac{1}{\sqrt{\lambda-1}}(\bfx+i\bfp).
\end{align}
By using the expansion~\eqref{eq:def-H-k-ell}, we can then write
\begin{align*}
\tr\big(\rhoth &C_{\bfr}^\dagger C_{\bfs}\big) \\
&= \sum_{ \mathclap{k, \ell, k', \ell'}}  \frac{\zeta_0^{k+k'}(\lambda)\zeta_1^{\ell+\ell'}(\lambda)}{\sqrt{k!\ell!k'!\ell'!}} r_1^k r_2^{\ell} s_1^{k'}s_2^{\ell'} \tr\big(\rhoth H^\dagger_{k, \ell} H_{k', \ell'}\big).
\end{align*}
Comparing this equation with~\eqref{eq:sigma-CC} we find that 
\begin{align}\label{eq:orth-rel-H}
\big\langle H_{k, \ell}, H_{k', \ell'}\big\rangle_{\rhoth}\! = \tr\big(\rhoth H_{k, \ell}^\dagger H_{k', \ell'}\big) = \delta_{k, k'}\delta_{\ell, \ell'}.
\end{align}

\begin{theorem}\label{thm:orth-basis}
Let $\rhoth$ be a single-mode thermal state with $\bfd(\rhoth)=0$ and $\bfgamma(\rhoth)= \lambda I$ with $\lambda> 1$. Define operators $H_{k, \ell}$ via~\eqref{eq:def-H-k-ell} where $C_{\bfr}$ is given in~\eqref{eq:def-C-r}. Then,  
\begin{align}\label{eq:basis-H-k-ell}
\Gamma=\{H_{k, \ell}:\, k, \ell\geq 0\},
\end{align} 
is an orthonormal basis for the space of operators with respect to the inner product $\langle \cdot, \cdot\rangle_{\rhoth}$.
\end{theorem}

\begin{proof}
We have already shown in~\eqref{eq:orth-rel-H} that $\{H_{k, \ell}:\, k, \ell\geq 0\}$ is an orthonormal set. It remains to show that $\{H_{k, \ell}:\, k, \ell\geq 0\}$ span the whole space of operators. 

Let $\mathcal V_t$ be the space of of polynomials of operators $\bfx, \bfp$ of degree \emph{at most} $t$. In other words,
$$\mathcal V_t = \text{span}\{I, \bfx, \bfp, \bfx^2, \bfx\bfp, \bfp^2, \dots, \bfx\bfp^{t-1},\bfp^t\}.$$
We note that, as mentioned above, $H_{k, \ell}$ is a polynomial of operators $\bfx, \bfp$ of degree $k+\ell$. Thus,  
$$\text{span}\{H_{k, \ell}:\, k+\ell\leq t\}\subseteq \mathcal V_t.$$ 
On the other hand, by the orthogonality relations already established, we know that $\{H_{k, \ell}:\, k+\ell\leq t\}$ is an independent set. Thus, computing the dimensions of both sides in the above inclusion, we find that $\text{span}\{H_{k, \ell}:\, k+\ell\leq t\}= \mathcal V_t$. Thus, taking union over $t\geq 1$, we find that $\text{span}(\Gamma)$ is equal to $\mathcal V_{\infty}=\cup_t \mathcal V_t$, i.e., the space of all polynomials of the quadrature operators. On the other hand, we know that the later set spans the whole space of operators.\footnote{This is essentially the content of the Stone-von Neumann theorem~\cite[Chapter 14]{Hall13}.} In fact, $\mathcal V_\infty$ contains the displacement operators, and any bounded operator can be expressed in terms of displacement operators~\cite{Cahill1969-I}.
\end{proof}

\section{Maximal Correlation for Bipartite Gaussian States}\label{sec:bipart-Gaussian}

As discussed, the maximal correlation is invariant under local unitary transformations. Hence, in order to compute the maximal correlation for bipartite Gaussian states, we can restrict to Gaussian states in the standard form~\eqref{eq:CM-StandForm} through Proposition~\ref{prop:standard-form}. Let $\rho_{AB}$ be a bipartite Gaussian state that is in the standard form with first moment $\bfd(\rho_{AB})=0$ and the covariance matrix
\begin{align}\label{eq:cov-matrix-mc}
\bfgamma_{AB} = \bfgamma(\rho_{AB}) = \begin{pmatrix}
\lambda_A I &  \bfnu\\
 \bfnu^\top & \lambda_B I
\end{pmatrix}, \quad  \boldsymbol{\nu}=\begin{pmatrix}
\nu_1 & 0\\
0 & \nu_2
\end{pmatrix}.
\end{align} 
Hence, the marginal states are thermal states with the covariance matrices $\bfgamma(\rho_A) =\lambda_A I$, $\quad \bfgamma(\rho_A) =\lambda_B I$ and $\bfd(\rho_A) = \bfd(\rho_B)=0$.

According to Theorem~\ref{thm:orth-basis}, we know that the sets 
$$\{H_{k, \ell}^A:\, k, \ell\geq 0\}
~\text{ and } 
~\{H_{k, \ell}^B:\, k, \ell\geq 0\},$$
are orthonormal bases with respect to the inner products $\langle \cdot, \cdot\rangle_{\rho_A}$ and $\langle \cdot, \cdot\rangle_{\rho_B}$, respectively. Here $H_{k, \ell}^A$'s and $H_{k, \ell}^B$'s are defined in terms of the corresponding quadrature operators and the parameter $\lambda_A$ and $\lambda_B$, respectively. To compute the maximal correlation $\mu(A, B)$, we use the above bases to expand the local operator $X_A$ and $Y_B$. Then, the objective function in the definition~\eqref{eq:def-max-cor}, which is equal to $|\langle X_A, Y_B\rangle_{\rho_{AB}}|$, can be expressed in terms of the inner products $\langle H_{k, \ell}^A, H_{k', \ell'}^B\rangle_{\rho_{AB}}$, which we compute in the following.

Let $C_{\bfr}^A$ and $C_{\bfr}^B$ be the operator~\eqref{eq:def-C-r} in terms of the modal quadrature operators $\bfR_A=(\bfx_A, \bfp_A)^{\top}$ and $\bfR_B=(\bfx_B, \bfp_B)^{\top}$, respectively.
As $C_{\bfr}^A$ and $C_{\bfs}^B$ are local operators and commute, we have
$$(C_\bfs^{A})^\dagger C_{\bfr}^B = \tau\, e^{\bfu^\top \widetilde \Upsilon \bfR_{AB}},$$
where 
$\tau=\exp\left({-\frac {\lambda_A}{4} \bfr^\top \bar \Upsilon \bar \Upsilon^\top \bfr-\frac {\lambda_B}{4} \bfs^\top  \Upsilon  \Upsilon^\top \bfs}\right),$
and
$$\bfu=\begin{pmatrix}
\bfr\\ \bfs
\end{pmatrix}, ~\text{ and } ~\widetilde \Upsilon = \begin{pmatrix}
\bar \Upsilon & 0\\ 
0 &  \Upsilon
\end{pmatrix}.$$
Thus, by using the Wigner function~\eqref{eq:gaussian-W}, we can compute the following inner product
\begin{align*}
\langle  C_{\bfr}^A, C_{\bfs}^B\rangle_{\rho_{AB}} &=   \tau \, \tr(\rho_{AB} e^{\bfu^\top \widetilde \Upsilon \bfR_{AB}})\\
 &= \tau \int_{\mathbb R^4} W_{\rho_{AB}}(\bfv) e^{ \bfu^\top \widetilde \Upsilon \bfv} \dd^4 \bfv\\
& = \tau\, e^{\frac 14 \bfu^\top \widetilde \Upsilon \bfgamma_{AB} \widetilde \Upsilon^\top \bfu}
 = e^{\frac 12 \bfr^\top \bar \Upsilon\bfnu  \Upsilon^\top \bfs}\\
& = \sum_{\mathclap{\substack{p_{00}, p_{01}\\ p_{10}, p_{11}}}}
\frac{\omega_{00}^{p_{00}} \omega_{01}^{p_{01}} \omega_{10}^{p_{10}} \omega_{11}^{p_{11}} }{p_{00}!p_{01}!p_{10}!p_{11}!}
r_1^{p_{00}+p_{01}} r_2^{p_{11}+p_{10}}\\ &\qquad\quad\times s_1^{p_{00}+p_{10}} s_2^{p_{01}+p_{11}}
\end{align*} 
where $\omega_{00}, \dots, \omega_{11}$ are entries of $\frac 12\bar \Upsilon\bfnu  \Upsilon^\top$ given by
$$
\begin{pmatrix}
\omega_{00} & \omega_{01}\\
\omega_{10} & \omega_{11}
\end{pmatrix} = \frac 12\bar \Upsilon\bfnu  \Upsilon^\top = \frac 14\begin{pmatrix}
\nu_1+\nu_2 & \nu_1-\nu_2\\
\nu_1-\nu_2 & \nu_1+\nu_2
\end{pmatrix}.
$$
On the other hand, using the expansion~\eqref{eq:def-H-k-ell}, we have 
\begin{align*}
\langle  C_{\bfs}^A, & C_{\bfr}^B\rangle_{\rho_{AB}}  \\
&= \sum_{\substack {k, \ell \\ k', \ell'}} \frac{\alpha_0^{k}\alpha_1^{\ell}\beta_0^{k'}\beta_1^{\ell'}}{\sqrt{k!\ell!k'!\ell'!}} r_1^{k}r_2^{\ell} s_1^{k'}s_2^{\ell'} \langle H_{k, \ell}^A, H_{k', \ell'}^B\rangle_{\rho_{AB}} ,
\end{align*}
where $\alpha_0=\zeta_0(\lambda_A)$, $\alpha_1=\zeta_1(\lambda_A)$,  $\beta_0=\zeta_0(\lambda_B)$ and $\beta_1=\zeta_1(\lambda_B)$ are given in~\eqref{eq:def-zeta}. Comparing the above equations 
yields
\begin{align}\label{eq:inner-prod-sum-zero}
\langle H_{k, \ell}^A, H_{k', \ell'}^B\rangle_{\rho_{AB}}=0, \quad \text{ if }\quad k+\ell\neq k'+\ell'.
\end{align}
Moreover, if $k+\ell = k'+\ell'$, then 
\begin{align*}
\langle  H_{k, \ell}^A,& H_{k', \ell'}^B\rangle_{\rho_{AB}} =  \frac{\sqrt{k!\ell!k'!\ell'!}}{\alpha_0^{k}\alpha_1^{\ell}\beta_0^{k'}\beta_1^{\ell'}}~~ \sum_{\mathclap{\substack{p_{00}+ p_{01}=k\\ p_{10}+p_{11}=\ell \\ p_{00}+ p_{10}=k'\\ p_{01}+p_{11}=\ell' } } } \quad
\frac{\omega_{00}^{p_{00}} \omega_{01}^{p_{01}} \omega_{10}^{p_{10}} \omega_{11}^{p_{11}} }{p_{00}!p_{01}!p_{10}!p_{11}!}.
\end{align*}

Now let $X_A, Y_B$ be arbitrary operators that satisfy the conditions $\tr(\rho_AX_A)=\tr(\rho_B Y_B)=0$ and $\tr(\rho_AX_A X_A^\dagger)=\tr(\rho_B Y_B Y^\dagger_B)=1$. Consider the expansion of these operators in the orthonormal bases $\{H_{k, \ell}^A:\, k, \ell\geq 0\}$ and $\{H_{k', \ell'}^B:\, k', \ell'\geq 0\}$ as follows:
$$X_A = \sum_{k, \ell} f_{k, \ell}   H_{k, \ell}^A, ~\text{ and }~ Y_B = \sum_{k', \ell'} g_{k, \ell}   H_{k', \ell'}^B.$$
The first condition on $X_A, Y_B$ means that $\langle I_A, X_A\rangle_{\rho_A}= \langle I_B, Y_B\rangle_{\rho_B}=0$. As $I_A=H_{0,0}^A $ and $I_B=H_{0,0}^B$, we find that $f_{0,0}=g_{0,0}=0$. The second condition on $X_A, Y_B$ can also be written as $\langle X_A, X_A\rangle_{\rho_A} =\langle Y_B, Y_B\rangle_{\rho_B}=1$. Thus, 
$$\|\bff\|^2= \|\bfg\|^2=1,$$
where $\bff{=} (f_{0,1}, f_{1, 0}, f_{1, 1}, \dots)^\top$ and $\bfg= (g_{0,1}, g_{1, 0}, g_{1, 1}, \dots)^\top$ are the vectors of coefficients excluding the first ones $f_{0,0}, g_{0,0}$ that are zero.
Hence, the inner product can be written as
\begin{equation}\label{eq:innerXY-expan}
\langle X_A, Y_B\rangle_{\rho_{AB}} = \sum_{\substack{k, \ell\\ k', \ell'}} \bar f_{k, \ell} g_{k', \ell'}  \langle H_{k, \ell}^A, H_{k', \ell'}^B\rangle_{\rho_{AB}}.
\end{equation}

Let $\widehat{\mathcal Q}$ be the matrix whose rows and column are indexed by pairs $\{(k, \ell):\, k+\ell\geq 1\}$, and whose $\big((k', \ell'), (k, \ell) \big)$-th entry is equal to   
\begin{align}\label{eq:def-matrix-Q}
\widehat{\cQ}_{(k', \ell'), (k, \ell) }=\langle H_{k, \ell}^A, H_{k', \ell'}^B\rangle_{\rho_{AB}}.
\end{align}
We note that by~\eqref{eq:inner-prod-sum-zero}, $\widehat{\cQ}$ is a block-diagonal matrix whose $t$-th block $\cQ^{(t)}$ is associated with pairs $(k, \ell)$ with $k+\ell=t$ and is of size $(t+1)\times (t+1)$:
\begin{align}\label{eq:Q-block}
\widehat\cQ= \begin{pmatrix}
\cQ^{(1)} & 0 & 0 & \cdots\\
0 & \cQ^{(2)} & 0 & \cdots\\
0 & 0 & \cQ^{(3)} & \cdots\\
\vdots & \vdots & \vdots & \ddots
\end{pmatrix}.
\end{align}
The matrix elements in the $t$-th block, for $0\leq \ell, \ell' \leq t$, are given by
\begin{align}
(\cQ^{(t)})_{\ell, \ell'} 
&\!= \langle H_{t-\ell, \ell}^A, H_{t-\ell', \ell'}^B\rangle_{\rho_{AB}}\nonumber\\
&\!= \frac{\sqrt{(t-\ell)!\ell!(t-\ell')!\ell'!}}{\alpha_0^{k}\alpha_1^{\ell}\beta_0^{k'}\beta_1^{\ell'}} \sum_{\mathclap{\substack{ p_{10}+p_{11}=\ell \\  p_{01}+p_{11}=\ell'\\ p_{00}+p_{01}+p_{10}+p_{11}=t } } } \quad
\!\frac{\omega_{00}^{p_{00}} \omega_{01}^{p_{01}} \omega_{10}^{p_{10}} \omega_{11}^{p_{11}} }{p_{00}!p_{01}!p_{10}!p_{11}!}. 
\label{eq:def-Q-t-ell}
\end{align}
In particular, 
the first block reads
\begin{align}\label{eq:cQ-1}
\cQ^{(1)} &
= \begin{pmatrix}
\frac{\omega_{00}}{\alpha_0\beta_0} & \frac{\omega_{01}}{\alpha_0\beta_1} \\
\frac{\omega_{10}}{\alpha_1\beta_0} & \frac{\omega_{11}}{\alpha_1\beta_1}
\end{pmatrix}\nonumber\\
&=\frac 12\! \begin{pmatrix}
\frac{(\nu_1+\nu_2)}{\sqrt{(\lambda_A+1)(\lambda_B+1) }} & \frac{(\nu_1-\nu_2)}{\sqrt{(\lambda_A+1)(\lambda_B-1) }}\\
\frac{(\nu_1-\nu_2)}{\sqrt{(\lambda_A-1)(\lambda_B+1) }} & \frac{(\nu_1+\nu_2)}{\sqrt{(\lambda_A-1)(\lambda_B-1) }}
\end{pmatrix}\nonumber\\
&=\frac 12 \begin{pmatrix}
\alpha_0^{-1} & 0\\
0 & \alpha_1^{-1}
\end{pmatrix}  \bar \Upsilon \bfnu \Upsilon^\top \begin{pmatrix}
\beta_0^{-1} & 0\\
0 & \beta_1^{-1}
\end{pmatrix}.
\end{align}

Now we can state the main result of this paper.

\begin{theorem}\label{thm:main-mc}
Let $\rho_{AB}$ be a Gaussian state with $\bfd(\rho_{AB})=0$ and covariance matrix~\eqref{eq:cov-matrix-mc}. Let $\{H_{k, \ell}^A:\, k, \ell\geq 0\}$ and $\{H_{k', \ell'}^B:\, k', \ell'\geq 0\}$ be orthonormal bases for the spaces of local operators of modes $A$ and $B$, respectively,  constructed 
in Theorem~\ref{thm:orth-basis}. Let $\widehat \cQ$ be the matrix consisting of inner products of operators in $\{H_{k, \ell}^A:\, k+ \ell> 0\}$ and $\{H_{k', \ell'}^B:\, k'+ \ell'> 0\}$ defined as in~\eqref{eq:def-matrix-Q} and with the block structure given in~\eqref{eq:Q-block} and~\eqref{eq:def-Q-t-ell}. Then, the maximal correlation for the Gaussian state is given by
$$\mu(A, B) = \|\cQ^{(1)}\|.$$
Equivalently, this means that in computing the maximal correlation in~\eqref{eq:def-max-cor} we may restrict to $X_A, Y_B$ that are linear in quadrature operators.
\end{theorem}

\begin{proof}
By using \eqref{eq:innerXY-expan} and \eqref{eq:Q-block}, the maximal correlation~\eqref{eq:def-max-cor} for a Gaussian state reads 
$$\mu(A, B) = \|\widehat\cQ\|=\max_{\|\bff\|=\|\bfg\|=1} |\bff^\dagger\widehat \cQ \bfg|,$$
where $ \|\widehat\cQ\|$ is the operator norm of $\widehat\cQ$. Given the block structure of $\widehat\cQ$ we know that $\|\widehat\cQ\| = \max_t \|\cQ^{(t)}\|$. Thus, to prove the theorem we need to show that 
$$\|\cQ^{(1)}\|\geq \|\cQ^{(t)}\|,\qquad \forall t>1.$$
To this end, we derive an equivalent representation of the matrices $\cQ^{(t)}$.

For $0\leq \ell\leq t$ and $\bfb=(b_1, \dots, b_t)\in \{0,1\}^t$ with $|\bfb|=\sum_i b_i$ define
\begin{align}\label{eq:def-S-matrix}
s_{\ell,\bfb} = \delta_{\ell, |\bfb|} \sqrt{\frac{\ell! (t-\ell)!}{t!}},
\end{align}
and let $S$ be the matrix of size $(t+1)\times 2^t$ with entries $s_{j, \bfb}$.
Observe that 
\begin{align*}
(SS^\dagger)_{\ell, \ell'} &= \sum_{\bfb} s_{\ell, \bfb} s_{\ell', \bfb}\\
& = \sum_{\bfb} \delta_{\ell, {|\bfb|}} \delta_{\ell', {|\bfb|}}  \frac{\ell! (t-\ell)!}{t!}\\
& = \delta_{\ell, \ell'}\sum_{\bfb: {|\bfb|}=\ell}  \frac{\ell! (t-\ell)!}{t!}\\
& =\delta_{\ell, \ell'}.
\end{align*}
Thus $SS^\dagger=I_{t+1}$ with $I_{t+1}$ being $(t+1)\times(t+1)$ identity matrix. Next, to simplify the notation we use 
$$\cQ=\cQ^{(1)},$$ 
and compute
\begin{align*}
(S&\cQ^{\otimes t} S^{\dagger})_{\ell, \ell'} \\
& = \sum_{\bfb, \bfb'} s_{\ell , \bfb} s_{\ell' , \bfb'} \cQ^{\otimes t}_{\bfb, \bfb'} \\
&= \sum_{\bfb, \bfb'}  \delta_{\ell, {|\bfb|}} \delta_{\ell', {|\bfb'|}} \frac{\sqrt{\ell! (t-\ell)!\ell'!(t-\ell')!} }{t!} \cQ^{\otimes t}_{\bfb, \bfb'}.
\end{align*} 
For a given $\bfb, \bfb'\in \{0,1\}^t$ define $p_{aa'}$ 
as follows:
\begin{align} \label{eq:def-j-aa'}
p_{aa'} = \{i:\, (b_i, b'_i) = (a,a')\}, \quad (a, a')\in \{0,1\}^2.
\end{align}
Hence, by using \eqref{eq:cQ-1}, we get
\begin{align*}
\cQ^{\otimes t}_{\bfb, \bfb'}& = \cQ_{00}^{p_{00}} \cQ_{01}^{p_{01}}\cQ_{10}^{p_{10}}\cQ_{11}^{p_{11}}\\
&= \frac{   \omega_{00}^{p_{00}} \omega_{01}^{p_{01}} \omega_{10}^{p_{10}} \omega_{11}^{p_{11}}    }{     \alpha_{0}^{p_{00}+p_{01}} \alpha_{1}^{p_{10}+p_{11}}   \beta_{0}^{p_{00}+p_{10}}     \beta_{1}^{p_{01}+p_{11}}    }.
\end{align*}
On the other hand, we note that if ${|\bfb|}=\ell$ and ${|\bfb'|}=\ell'$, then 
\begin{equation}
\ell=p_{10}+p_{11}, ~ \text{ and } ~ \ell'=p_{01}+p_{11}.
\end{equation}
For a fixed tuple $(p_{00}, p_{01}, p_{10}, p_{11})$ satisfying these equations, we can see that there are 
$$\binom{t}{p_{00}, p_{01}, p_{10}, p_{11}} = \frac{t!}{p_{00}! p_{01}! p_{10}! p_{11}!},$$
pairs of $(\bfb, \bfb')$ that satisfy~\eqref{eq:def-j-aa'}. Therefore, putting these together and comparing with~\eqref{eq:def-Q-t-ell}, we obtain
\begin{widetext}
\begin{align*}
(S\cQ^{\otimes t} S^{\dagger})_{\ell, \ell'} & = ~~\sum_{\mathclap{\substack{p_{10}+p_{11}=\ell\\ p_{01}+p_{11}=\ell'\\ p_{00}+p_{01}+p_{10}+p_{11}=t}}} ~~\binom{t}{p_{00}, p_{01}, p_{10}, p_{11}}  \frac{\sqrt{\ell! (t-\ell)!\ell'!(t-\ell')!} }{t!} \cQ_{11}^{p_{00}} \cQ_{12}^{p_{01}}\cQ_{21}^{p_{10}}\cQ_{22}^{p_{11}}\\
& ~~= \sum_{\mathclap{\substack{p_{10}+p_{11}=\ell\\ p_{01}+p_{11}=\ell'\\ p_{00}+p_{01}+p_{10}+p_{11}=t}}} ~~ \frac{\sqrt{\ell! (t-\ell)!\ell'!(t-\ell')!} }{p_{00}! p_{01}! p_{10}! p_{11}!} \frac{   \omega_{00}^{p_{00}} \omega_{01}^{p_{01}} \omega_{10}^{p_{10}} \omega_{11}^{p_{11}}    }{     \alpha_{0}^{p_{00}+p_{01}} \alpha_{1}^{p_{10}+p_{11}}   \beta_{0}^{p_{00}+p_{10}}     \beta_{1}^{p_{01}+p_{11}}    }\\
& = \cQ^{(t)}_{\ell, \ell'}.
\end{align*}
\end{widetext}
Therefore, $S\cQ^{\otimes t} S^{\dagger}= \cQ^{(t)}$ and
\begin{align*}
\|\cQ^{(t)}\| &=\|S\cQ^{\otimes t} S^{\dagger}\|\\
& \leq \|S\|\cdot \|S^{\dagger}\|\cdot \|\cQ^{\otimes t}\|\\
& =  \|\cQ^{\otimes t}\| = \|\cQ\|^t,
\end{align*}
where in the second line we use the fact that $SS^\dagger=I_{t+1}$ which gives $\|S\|=\|S^\dagger\|= 1$. Next, we note that $\|\cQ\|\leq 1$ because
\begin{align*}
\|\cQ\| & =\max_{\bff, \bfg\neq 0} \frac{| \bff^\dagger \cQ \bfg |}{\|\bff\|\cdot \|\bfg\|}\\
& = \max_{\bff, \bfg\neq 0} \frac{| \langle f_0H_{1,0}^A + f_1H_{0,1}^A, g_0H_{1,0}^B + g_1H_{0,1}^B\rangle_{\rho_{AB}}  |}{\|\bff\|\cdot \|\bfg\|}\\
& \leq  \max_{\bff, \bfg\neq 0} \frac{\|  f_0H_{1,0}^A + f_1H_{0,1}^A \|\cdot \|g_0H_{1,0}^B + g_1H_{0,1}^B \|}{\|\bff\|\cdot \|\bfg\|}\\
&= 1,
\end{align*}
where in the third line we use the Cauchy-Schwarz inequality and in the last line we use the fact that $\{H_{1,0}^A + H_{0,1}^A\}$ and $\{H_{1,0}^B + H_{0,1}^B\}$ are orthonormal. Thus,
$$\|\cQ^{(t)}\|  \leq \|\cQ\|^t \leq \|\cQ\|= \|\cQ^{(1)}\|,$$
and 
therefore the maximal correlation for a Gaussian state in the standard form becomes
\begin{align}\label{eq:MC-Gaussian-SF}
\mu(A, B) &=\|\cQ^{(1)}\|=\max_{\|\bff\|=\|\bfg\|=1} | \bff^\dagger \cQ^{(1)} \bfg |.
\end{align}
This implies that the optimal local operators by using~\eqref{eq:H-t=1} can be expressed as a linear combination of quadrature operators
\begin{align}
X_A&=f_0H_{1,0}^A + f_1H_{0,1}^A = \alpha_x \bfx_A + \alpha_p \bfp_A,\label{eq:X-optimal}\\
Y_B&=g_0H_{1,0}^B + g_1H_{0,1}^B = \beta_x \bfx_B + \beta_p \bfp_B,\label{eq:Y-optimal}
\end{align}
where the coefficients are determined by~\eqref{eq:MC-Gaussian-SF}. Note that, in general, the coefficients may be complex numbers, and hence the optimal operators may not correspond to local physical observables.
\end{proof}

\section{Gaussian maximal correlation}\label{sec:GMC}
 
In this section, we define another measure of correlation for bipartite Gaussian states by restricting the optimization in \eqref{eq:def-max-cor} to Gaussian observables~\cite{Holevo2019}, i.e., operators that are Hermitian and linear in terms of quadrature operators and correspond to homodyne measurements. We refer to this new correlation measure as the \emph{Gaussian maximal correlation} which is given by
\begin{equation}
\mu_{G}(A, B) = \sup_{X_A, Y_B} \tr(\rho_{AB} X_A\otimes Y_B),
\end{equation}
where $\rho_{AB}$ is a bipartite Gaussian state, and $X_A, Y_B$ are local Hermitian observables that are linear in terms of quadrature operators and satisfy $\tr(\rho_A X_A)= \tr(\rho_B Y_B)=0$ and $\tr(\rho_AX_A^2)=\tr(\rho_B Y_B^2)=1$. Note that by definition, $\mu_{G}(A, B)\leq \mu(A, B)$.

Suppose that the covariance matrix of $\rho_{AB}$ is
\begin{align*}
	\bfgamma_{AB} = \bfgamma(\rho_{AB}) = \begin{pmatrix}
		\bfgamma_A  &  \bfnu_{AB}\\
		\bfnu_{AB}^\top & \bfgamma_B
	\end{pmatrix},
\end{align*}
and for simplicity assume that $\bfd(\rho_{AB})=0$. By the restrictions on $X_A, Y_B$, there are real vectors $\bfr_A, \bfr_B$ such that $X_A=\bfr_A^\top \bfR_A$, $Y_B=\bfr_B^\top \bfR_B$ and 
\begin{align*} 
	1&=\tr\big(\rho_A X_A^2\big)\\
	&=\frac 1 2 \tr\big( \rho_A \{X_A, X_A\}\big)\\
	&=\frac 1 2 \bfr_A^\top \tr\big( \rho_A \{\bfR_A, \bfR_A\}\big)\bfr_A\\
	&=\frac 1 2 \bfr_A^\top \bfgamma_A \bfr_A,
\end{align*}
and similarly $\bfr_B^\top \bfgamma_B \bfr_B =2$. Note that $\tr(\rho_A X_A) = \tr(\rho_B Y_B)=0$ is automatically satisfied since $\bfd(\rho_{AB})=0$.
Next, similar to the above computation, it is easily verified that $ 2\tr(\sigma_{AB} X_A\otimes Y_B)=\bfr_A^\top \bfnu_{AB} \bfr_B$. Therefore, rescaling $\bfr_A, \bfr_B$ by a factor of $\sqrt{2}$, the Gaussian maximal correlation reads
\begin{align}\label{eq:GMC}
	\mu_{G}(A, B) = \max_{\bfr_A, \bfr_B}~ &~ \bfr_A^\top \bfnu_{AB} \bfr_B\\
	& ~\bfr_A^\top \bfgamma_A \bfr_A=1\nonumber\\
	& ~\bfr_B^\top \bfgamma_B \bfr_B=1.\nonumber
\end{align}
Writing the above optimization in terms of $\bfs_A= \bfgamma_A^{1/2}\bfr_A$ and $\bfs_B= \bfgamma_B^{1/2}\bfr_B$, we find that 
\begin{align}\label{eq:mu-G-norm}
\mu_G(A, B)=\|\bfgamma_A^{-1/2}\bfnu_{AB}\bfgamma_B^{-1/2}\|.
\end{align}
Therefore, for a Gaussian state in the standard form \eqref{eq:cov-matrix-mc} with $|\nu_1|\geq|\nu_2|$, which can always be accomplished by relabeling the phase-space quadratures, the Gaussian maximal correlation becomes
\begin{align}\label{eq:GMC-SF}
	\mu_{G}(A, B) = \frac{|\nu_1|}{\sqrt{\lambda_A \lambda_B}}.
\end{align}
Note that the joint probability distribution $p(x_A,x_B)$ associated with local homodyne measurements on the $x$-quadratures is Gaussian with $\text{Cov}[x_A,x_B]=\nu_1$,  $\text{Var}[x_A]=\lambda_A$, and $\text{Var}[x_B]=\lambda_B$. Therefore, comparing this equation with \eqref{eq:ClasMC-Gaus}, we observe that the Gaussian maximal correlation is in fact the classical maximal correlation of the outcome probability distribution of optimal local homodyne measurements. However, unlike the classical case, \eqref{eq:GMC-SF} shows that  $\mu_{G}(A, B)$ cannot be equal to one, as due to the uncertainty relation~\eqref{eq:uncertainty} the covariance matrix of physical states cannot have a zero eigenvalue.

Considering the characterization of Gaussian cptp maps~\cite{Serafini17}, one can verify that $\mu_{G}(A, B)$ is monotone under local Gaussian cptp maps; the point is that under such local super-operators in the Heisenberg picture, linear operators are mapped to linear operators. Therefore, correlations between phase-space quadratures cannot be increased under Gaussian operations. Gaussian maximal correlation also satisfies the tensorization property. To verify this, it suffices to use~\eqref{eq:mu-G-norm}, the fact that $\bfgamma(\rho_{AB}\otimes \rho_{A'B'}) = \bfgamma(\rho_{AB})\oplus \bfgamma(\rho_{A'B'})$ and $\|\boldsymbol{\xi}\oplus \boldsymbol{\eta}\|= \max\{\|\boldsymbol{\xi}\|, \|\boldsymbol{\eta}\|\}$. Therefore, the Gaussian maximal correlation satisfies both the monotonicity and tensorization properties, and therefore can be used to study the local state transformation problem when both resource and target states are Gaussian. For example, this result shows that $n$ copies of a Gaussian in the standard form with the Gaussian maximal correlation~\eqref{eq:GMC-SF}, cannot be locally transformed into another Gaussian state with the same marginal states but with $|\nu_1'|>|\nu_1|$.

We state yet another reformulation of the Gaussian maximal correlation. Note that 
\begin{align*}
(\bfgamma_{A}\oplus \bfgamma_{B})^{-1/2} &\bfgamma_{AB} (\bfgamma_{A}\oplus \bfgamma_{B})^{-1/2} \\
&= \begin{pmatrix}
I_A & \bfgamma_A^{-1/2}\bfnu_{AB}\bfgamma_B^{-1/2}\\
\bfgamma_B^{-1/2}\bfnu_{AB}^\top \bfgamma_A^{-1/2} & I_B
\end{pmatrix}.
\end{align*}
Then, using~\eqref{eq:mu-G-norm} we fine that $(\bfgamma_{A}\oplus \bfgamma_{B})^{-1/2} \bfgamma_{AB} (\bfgamma_{A}\oplus \bfgamma_{B})^{-1/2}\geq 1- \mu_G(A, B)$. Therefore, we can write
$\mu_{G}(A, B)=1-V(A,B)$ where $V(A,B)=\max\{q\leq1:\bfgamma_{AB}\geq q(\bfgamma_{A}\oplus \bfgamma_{B})\}$. This  formulation of the Gaussian maximal correlation is related to the measure for Gaussian resources~\cite{LamiPRA2018} and is reminiscent of the entanglement measure introduced in~\cite{GiedkeCirac02} for Gaussian states.
Indeed, the measure of~\cite{GiedkeCirac02} is defined by $V'(A, B)=\max\{q\leq1:\bfgamma_{AB}\geq q(\bfgamma'_{A}\oplus \bfgamma'_{B})\}$ where $\bfgamma'_A, \bfgamma'_B$ are covariance matrices of \emph{some} quantum states, while in our case they are the covariance matrices of the marginal states of $\rho_{AB}$. 

We finish this section by emphasizing on the fact that the Gaussian maximal correlation is monotone only if we restrict to local Gaussian operations, while the quantum maximal correlation studies in previous sections is monotone under all local operations.  

\section{Examples}\label{sec:Examples}
In this section, we further illustrate various features of the correlation measures by considering examples of Gaussian states. We consider two classes of bipartite Gaussian states that are in the standard form with symmetric marginal states and zero first-order moments.

The first class of Gaussian states is described by the covariance matrix \eqref{eq:CM-StandForm} with $\lambda_A=\lambda_B=\lambda>1$ and $\nu_1=-\nu_2=\nu\geq0$. The physicality condition~\eqref{eq:uncertainty} implies that $\nu\leq\sqrt{\lambda^2-1}$, and the state is entangled for $\lambda-1<\nu<\sqrt{\lambda^2-1}$. We refer to these states as correlated-anticorrelated (CA) Gaussian states. The maximal correlation for these states, using \eqref{eq:cQ-1} and \eqref{eq:MC-Gaussian-SF}, is given by
\begin{align}
\mu_{_\text{CA}}(A,B)&= \max_{\|\bff\|=\|\bfg\|=1}  \begin{pmatrix}\bar{f}_0 & \bar{f}_1 \end{pmatrix}\!\!
\begin{pmatrix}
	0& \frac{\nu}{\sqrt{\lambda^2-1 }}\\
	\frac{\nu}{\sqrt{\lambda^2-1 }} & 0
\end{pmatrix}\!\!
\begin{pmatrix}g_0 \\ g_1 \end{pmatrix}\nonumber\\
&=\frac{\nu}{\sqrt{\lambda^2-1}},\label{eq:MC-CA}
\end{align} 
where the maximum is attained for $f_0=f_1=g_0=g_1=1/\sqrt{2}$. Hence, optimal local operators using \eqref{eq:X-optimal} and \eqref{eq:Y-optimal} can be found as $X_A=\bfa_A^\dagger/\sqrt{\lambda+1}+\bfa_A/\sqrt{\lambda-1}$ and $Y_A=\bfa_B^\dagger/\sqrt{\lambda+1}+\bfa_B/\sqrt{\lambda-1}$, which are not hermitian and cannot be viewed as physical observables.  
Note that for $\nu=\sqrt{\lambda^2-1}$, the Gaussian state corresponds to a two-mode squeezed vacuum state,
\begin{equation*}
\ket{\psi_{AB}}=\sqrt{\frac{2}{\lambda+1}}\sum_n \bigg(\frac{\lambda-1}{\lambda+1}\bigg)^{\!\!{n}/{2}}\! \ket{n_A,n_B},
\end{equation*}
which is a \emph{pure entangled state} with $\mu_{_\text{CA}}(A,B)=1$. Here $\{\ket{n_A}\}$ and  $\{\ket{n_B}\}$ are the Fock states. Thus, following Section~\ref{sec:MC}, one can easily verify that for two-mode squeezed vacuum states the maximal correlation of $\mu=1$ can be achieved using local hermitian operators $X_A=(\bfa_A^\dagger\bfa_A-\bar{n})/\Delta n=H_{1,1}^A$ and $Y_B=(\bfa_B^\dagger\bfa_B-\bar{n})/\Delta n=H_{1,1}^B$, where $\bar{n}=(\lambda-1)/2$ and $\Delta n=\sqrt{\lambda^2-1}/2$ are the mean and the standard deviation of photon-number distribution of the marginal states. This implies that the same maximal correlation can be obtained for two-mode squeezed vacuum states by using local \emph{hermitian} operators that are \emph{quadratic} in the phase-space quadrature operators. This means that, in general, optimal local operators for the maximal correlation are not unique.

As discussed, the maximal correlation for Gaussian states can be used to study the local state transformation problem when the target state is not Gaussian. For example, by using \eqref{eq:MC-CA}, it can be seen that any arbitrary number of copies of CA Gaussian states cannot be locally transformed into the noisy Bell state~\eqref{eq:Werner state} without classical communication if $\nu<\kappa\sqrt{\lambda^2-1}$. Note, however, that if two states have the same amount of the maximal correlation, it is not clear whether or how the two states can locally be transformed into one another in general. 

We also consider a second class of Gaussian states described by the covariance matrix \eqref{eq:CM-StandForm} with $\lambda_A=\lambda_B=\lambda>1$ and $\nu_1=\nu_2=\nu\geq0$. Note that for all physical values of $\nu\leq \lambda-1$, these states are separable with a nonnegative Glauber-Sudarshan $P$-function. We refer to these states as correlated-correlated (CC) Gaussian states. Using~\eqref{eq:cQ-1} and \eqref{eq:MC-Gaussian-SF}, the maximal correlation is given by
\begin{align}
	\mu_{_\text{CC}}(A,B)&= \max_{\|\bff\|=\|\bfg\|=1}  \begin{pmatrix}\bar{f}_0 & \bar{f}_1 \end{pmatrix}\!\!
	\begin{pmatrix}
		\frac{\nu}{\lambda+1}& \\
	0 & 	\frac{\nu}{\lambda-1}
	\end{pmatrix}\!\!
	\begin{pmatrix}g_0 \\ g_1 \end{pmatrix}\nonumber\\
	&=\frac{\nu}{\lambda-1}.\label{eq:MC-CC}
\end{align} 
Here, for $g_0=f_0=0$ and  $g_1=f_1=1$, we can see that optimal local operators $X_A=\bfa_A\sqrt{2/(\lambda-1)}$ and $Y_B=\bfa_B/\sqrt{2/(\lambda-1)}$ are not hermitian again. Of particular interest is the case of $\nu=\lambda-1$ where $\mu_{_\text{CC}}(A,B)=1$. In this case, the Gaussian state can be written as
\begin{equation*}
	\rho_{AB}=\int \frac{2e^{-2|\alpha|^2/(\lambda-1)}}{(\lambda-1)\pi}  \ketbra{\alpha}{\alpha} \otimes \ketbra{\alpha}{\alpha}\dd^{2}\alpha.
\end{equation*}
Using this representation, we can see that for any local measurements $\{\Pi_{A,n}: n\}$ and  $\{\Pi_{B,m}: m\}$, 
\begin{align}
	\begin{split}
\tr\big(\rho_{AB} \Pi_{A,n}\!\otimes \Pi_{B,m}\big)\!=\!\! \int & \frac{2e^{-2|\alpha|^2/(\lambda-1)}}{(\lambda-1)\pi}  \bra{\alpha}\Pi_{A,n}\ket{\alpha} \\
&\times \bra{\alpha}\Pi_{B,m}\ket{\alpha}\dd^{2}\alpha \neq0,
	\end{split}\label{eq:CC-joint-prob}
\end{align}
for any $n$ and $m$. 
Note that here the positivity of measurement operators implies $\bra{\alpha}\Pi_{B,n}\ket{\alpha}\geq0$ and $\bra{\alpha}\Pi_{A,m}\ket{\alpha}\geq0$. Also, $\bra{\alpha}\Pi\ket{\alpha}$ is, in general, an everywhere convergent power series in terms of $\alpha$ and $\bar{\alpha}$ and cannot be identical to zero on some nontrivial region, unless $\Pi=0$~\cite{Cahill1965}. Therefore, $\bra{\alpha}\Pi_{A,n}\ket{\alpha}\bra{\alpha}\Pi_{B,m}\ket{\alpha}$ cannot be zero almost everywhere. As the joint outcome probability distribution~\eqref{eq:CC-joint-prob} is not in the form of~\eqref{eq:perfect-shar-rand}, $\mu_{_\text{CC}}(A,B)=1$ for $\nu=\lambda-1$ cannot be achieved using hermitian local operators. This example shows that the maximal correlation of continuous-variable systems must be optimized over all hermitian and non-hermitian local operators.

It is easy to verify that even for non-Gaussian bipartite states of this form
\begin{equation}
\rho_{AB}=\int P(\alpha) \ketbra{\alpha}{\alpha} \otimes \ketbra{\alpha}{\alpha} \dd^{2}\alpha,
\end{equation}
where $P(\alpha)$ is not necessarily Gaussian, the quantum maximal correlation of one is achieved using non-hermitian local operators $X_A=(\bfa_{A}-\langle\bfa_{A}\rangle)/\Delta\bfa_{A}$ and $Y_B=(\bfa_{B}-\langle\bfa_{B}\rangle)/\Delta\bfa_{B}$ with $\langle\bfa_{A}\rangle=\langle\bfa_{A}\rangle=\int P(\alpha)\alpha\, \dd^{2}\alpha$ and $\Delta\bfa_{A}^2=\int P(\alpha)|\alpha|^2 \dd^{2}\alpha-\langle\bfa_{A}\rangle^2=\Delta\bfa_{B}^2$. This fact can also be verified by examples of two-qubit states; see the updated arXiv version of~\cite{Beigi13}.

Comparing the maximal correlation of the CC and CA states for the same marginal states and $0\leq\nu\leq\lambda-1$, in which case both states are separable, we can see that $\mu_{_\text{CC}}(A,B)<\mu_{_\text{CA}}(A,B)$. Although the only difference between the two classes is the sign of the correlation, this shows that any arbitrary number of copies of the CC state cannot be  transformed to the CA state by local operations and without classical communication. We note, however, that these two states have the same amount of Gaussian maximal correlation since by~\eqref{eq:GMC-SF}  
\begin{equation*}
	\mu_{\text{G}_\text{CA}}(A,B)=\mu_{\text{G}_\text{CC}}(A,B)=\frac{\nu}{\lambda}.
\end{equation*}
Hence, although the maximal correlation gives impossibility of local state transformation of the CC state to the CA state, the Gaussian maximal correlation cannot detect this. We also note that for these states the Gaussian maximal correlation is strictly less than the maximal correlation given by~\eqref{eq:MC-CA} and~\eqref{eq:MC-CC}. 

As an application of the Gaussian maximal correlation measure, suppose that a third party prepares a CA state $\rho_{AB}$ and sends one subsystem to Alice and the other one to Bob through lossy channels. The covariance matrix of the shared state between Alice and Bob is given by~\cite{Serafini17}
\begin{equation*}
	\begin{pmatrix}
		(\tau_A\lambda+1-\tau_A) I&\sqrt{\tau_A\tau_B}\nu Z \\
		\sqrt{\tau_A\tau_B}\nu Z& (\tau_B\lambda+1-\tau_B)I\\ 
	\end{pmatrix},
\end{equation*}
where $0\leq\tau_A\leq1$ and $0\leq\tau_B\leq1$ are the transmissivities of lossy channels to Alice and Bob, and $Z=\text{diag}(1,-1)$.  We can then see that, if $\tau_A<1$ and/or $\tau_B<1$, the Gaussian maximal correlation of the shared state is less than the Gaussian maximal correlation of the initial state,
\begin{equation*}
	\frac{\sqrt{\tau_A\tau_B}\nu}{\sqrt{(\tau_A\lambda+1-\tau_A)(\tau_B\lambda+1-\tau_B)}}<\frac{\nu}{\lambda}.
\end{equation*}
Therefore, using any arbitrary number of copies of the shared state and without classical communication, Alice and Bob cannot locally retrieve the initial state $\rho_{AB}$.


\section{Multipartite Gaussian States}\label{sec:MC-ribbon-Gaussian}

Classical maximal correlation for multipartite probability distributions is first defined in~\cite{BeigiGohari18}. The maximal correlation of an $m$-partite distribution is a \emph{subset} of $[0,1]^m$, called the \emph{maximal correlation ribbon}. This subset for $m=2$ is fully characterized in terms of  a single number that is the (bipartite) maximal correlation~\cite[Proposition~29]{BeigiGohari18}. Thus, the maximal correlation ribbon is really a generalization of the (bipartite) maximal correlation. Moreover, maximal correlation ribbon satisfies the data processing and tensorization properties, which are required for the local state transformation problem in the multipartite case.

In this section, we first generalize the classical notion of the maximal correlation ribbon for multipartite quantum states. Then, we characterize the quantum maximal correlation ribbon for multipartite Gaussian states. Here, we briefly discuss these results and give the details in  Appendix~\ref{app:MCRibbon} and Appendix~\ref{app:Gaussian}.

Following the definition of the maximal correlation ribbon in the classical case~\cite{BeigiGohari18}, in order to define a quantum maximal correlation ribbon we first need a notion of \emph{quantum conditional expectation}. Let $\rho_{A^m}=\rho_{A_1, \dots, A_m}$ be an $m$-partite quantum state, and let $\rho_{A_j}$, $j=1, \dots, m$ be its marginal states. For simplicity we assume that $\rho_{A_j}$ is full-rank for all $j$. Then, for any $1\leq j\leq m$ define the super-operator $\cE(\cdot)$ by
$$\cE_j(X) = \tr_{\neg j} (X\rho_{A^m} \rho_{A_j}^{-1}),$$
where by $\tr_{\neg j}$ we mean tracing out all the subsystems except the $j$-th one. $\cE_j(\cdot)$ behaves like a ``conditional expectation" whose properties are given in Lemma~\ref{lem:cE} in Appendix~\ref{app:MCRibbon}.

Next, recall that the variance of an operator $X$ is defined by
\begin{align*}
	\Var(X) &= \langle X, X\rangle_{\rho_{A^m}} - |\tr(\rho_{A^m}X)|^2\\
	&= \langle X-\tr(\rho_{A^m}X) I, X-\tr(\rho_{A^m}X) I\rangle_{\rho_{A^m}}\\ 
	&= \|X-\tr(\rho_{A^m}X) I\|^2,
\end{align*}
where as before $\langle X, Y\rangle_{\rho_{A^m}} = \tr(\rho_{A^m} X^\dagger Y)$.

\begin{definition}
	For an $m$-partite quantum state $\rho_{A_1, \dots, A_m}$ let $\fS(A_1, \dots, A_m)\subseteq [0,1]^m$ be the set of tuples $(\theta_1, \dots, \theta_m)$ such that for any operator $X=X_{A^n}$ we have
	\begin{align} \label{eq:def-MCR}
		\Var(X)\geq \sum_j \theta_j \Var(\cE_j(X)).
	\end{align}
	We call  $\fS(A_1, \dots, A_m)$ the \emph{maximal correlation ribbon} (MC ribbon) of $\rho_{A_1, \dots, A_m}$.
\end{definition}

Remark that since $\Var(X)\geq 0$, in the definition of $\fS(A_1, \dots, A_m)$ we restrict to $\theta_j\geq 0$. Moreover, as shown in Lemma~\ref{lem:cE} in Appendix~\ref{app:MCRibbon}, we have 
\begin{align}\label{eq:Var-E-Var-ineq}
\Var(X)\geq \Var(\cE_j(X)), \quad \forall j.
\end{align}
Thus,~\eqref{eq:def-MCR} implies $\theta_j\leq 1$. This is why $\fS(A_1, \dots, A_m)$ is defined as a subset of $[0,1]^m$. Moreover, by~\eqref{eq:Var-E-Var-ineq}, any $(\theta_1, \dots, \theta_m)\in [0,1]^m$ with $\sum_j \theta_j\leq 1$ belongs to $\fS(A_1, \dots, A_m)$ for any $\rho_{A^m}$. Thus, the interesting part of $\fS(A_1, \dots, A_m)$ is the subset $[0,1]^m$ beyond the above trivial subset. Finally, from the definition, it is clear that $\fS(A_1, \dots, A_m)$ is a convex set. See Figure~\ref{fig:MCRibbon} for a typical shape of the MC~ribbon. 

\begin{figure}
	\begin{center}
		\includegraphics[scale=0.42]{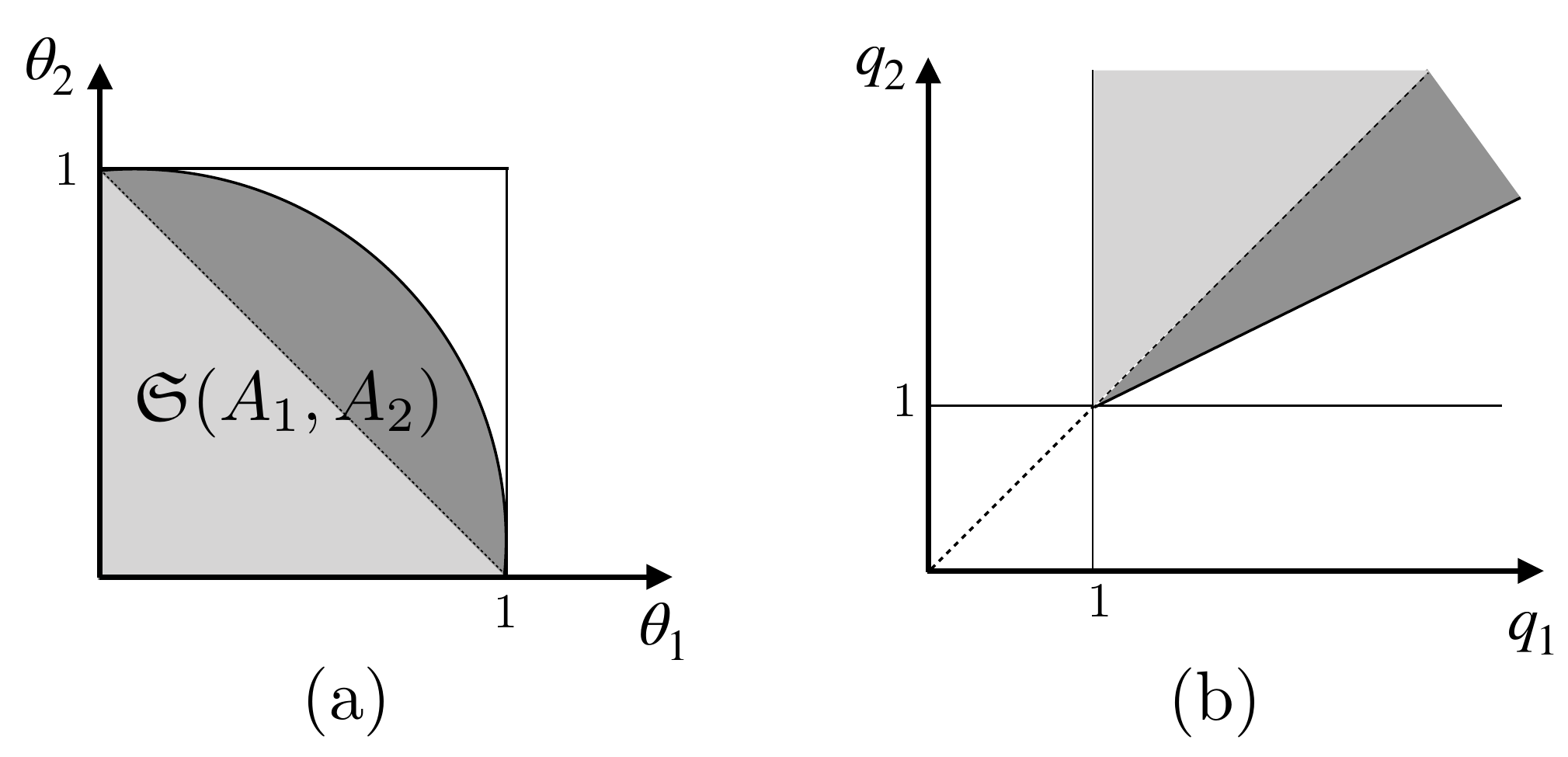}
		\caption{\footnotesize (a) The MC~ribbon is a convex subset of $[0,1]^m$, and the non-trivial part of $\fS(A_1, \dots, A_m)$ is the set of tuples $(\theta_1, \dots, \theta_m)$ with $\sum_j \theta_j>1$. For $m=2$, the MC ribbon of a generic state is the dark gray area plus the light gray area that is the trivial part, containing points satisfying $\theta_1+\theta_2\leq1$. We note that by equation~\eqref{eq:MCR-m2-q} when $\mu(A_1, A_2)=1$, the MC ribbon is equal to this trivial part. (b) In the case of $m=2$, with the re-parametrization $q_1= 1/(1-\theta_1)$ and $q_2=1/\theta_2$, by equation~\eqref{eq:MCR-m2-q}, the MC~ribbon is determined by the line $\mu^2(A_1, A_2)(q_1-1) = (q_2-1)$. We note that in this case the non-trivial part of the MC~ribbon takes the form a ribbon, thus the name.
		}
		\label{fig:MCRibbon}
	\end{center}
\end{figure}  

It can be shown that for product states $\rho_{A^m}= \rho_{A_1}\otimes \cdots \otimes \rho_{A_m}$, the MC~ribbon is the largest possible set: $\fS(A_1, \dots, A_m)= [0,1]^m$. On the other hand, for the maximally entangled state $\frac{1}{\sqrt 2}(\ket{00} +\ket{11})$, the MC~ribbon consists of only trivial points: $\fS(A_1, A_2)= \{(\theta_1, \theta_2)\in [0,1]^2:\,  \sum_j \theta_j\leq 1 \}$. See Appendix~\ref{app:MCRibbon} for more details.

In particular, as shown in Theorem~\ref{thm:MCR-bipartite} in Appendix~\ref{app:MCRibbon} when $m=2$, the pair $(\theta_1, \theta_2)\in [0,1]^2$ belongs to $\fS(A_1, A_2)$ if and only if
\begin{align}\label{eq:MCR-m2-q}
	\big(\frac{1}{\theta_1} -1 \big)\big(\frac{1}{\theta_2} -1 \big)\geq \mu^2(A_1, A_2),
\end{align}
where $\mu(A_1, A_2)$ is the maximal correlation of the bipartite state. This shows that the  MC~ribbon can be viewed as the generalization of the maximal correlation.

As shown in Theorem~\ref{thm:MCR-monotonicity-ternsorization} in Appendix~\ref{app:MCRibbon}, maximal correlation ribbon also satisfies the data processing and the tensorization properties:
\begin{itemize}
	\item \textbf{(Data processing)} For local quantum operations $\Phi^{(j)}_{A_j\to A'_j}$, if
	$$\sigma_{A'^m} = \Phi^{(1)}_{A_1\to A'_1}\otimes\cdots \otimes \Phi^{(m)}_{A_m\to A'_m}(\rho_{A^m}),$$
	then $\fS(A_1, \dots, A_m)\subseteq \fS(A'_1, \dots, A'_m)$.
	
	\item \textbf{(Tensorization)} For any $\rho_{A^m}$ and any integer $n$ we have
	$$\fS(A_1^n,\dots, A_m^n) = \fS(A_1,\dots,  A_m),$$
	where the left hand side is the MC ribbon of $\rho_{A^m}^{\otimes n}$.
\end{itemize}

These two properties show that the MC~ribbon is a relevant invariant for the local state transformation problem in the multipartite case. 

\medskip
Now we get to our main problem, namely, how to compute the MC~ribbon for multipartite Gaussian states. The following theorem, that is a generalization of Theorem~\ref{thm:main-mc}, is our main result in this direction.

\begin{theorem} (Informal)
	To compute the MC~ribbon for Gaussian state it suffices to restrict $X$ in~\eqref{eq:def-MCR} to operators that are linear in phase-space quadrature operators.
\end{theorem}

For a formal statement of this result and its proof see Appendix~\ref{app:Gaussian}. 

Motivated by the definition of the Gaussian maximal correlation~\eqref{eq:GMC}, we can define a Gaussian maximal correlation ribbon for multipartite Gaussian states as follows. For an $m$-mode Gaussian state $\rho_{A_1, \dots, A_m}$ let $\fS_{G}(A_1, \dots, A_m)$ be the set of tuples $(\theta_1, \dots, \theta_m)\in [0,1]^m$ satisfying 
	$$\Var\big(\sum_{j=1}^m X_j\big) \leq \sum_{j=1}^m \frac{1}{\theta_j} \Var(X_j),$$
	for all local operators $X_1, \dots, X_m$ that are hermitian and linear in terms of phase-space quadrature operators. We note that, by Theorem~\ref{thm:MCR-equiv} in Appendix~\ref{app:MCRibbon}, removing the latter constraints on $X_j$'s, we recover $\fS(A_1, \dots, A_m)$. Therefore, we have $\fS(A_1, \dots, A_m)\subseteq \fS_{G}(A_1, \dots, A_m)$. Moreover, it is not hard to verify that $\fS_{G}(A_1, \dots, A_m)$ is monotone under the action of local Gaussian operations and satisfies tensorization (in the sense of Theorem~\ref{thm:MCR-monotonicity-ternsorization}).

To find an equivalent characterization of $\fS_{G}(A_1, \dots, A_m)$,
	a straightforward computation as in the bipartite case shows that $(\theta_1, \dots, \theta_m)$ belongs to $\fS_{G}(A_1, \dots, A_m)$ if and only if 
	\begin{align*}
		\begin{pmatrix}
			\frac{1}{\theta_1} \bfgamma_{A_1} & 0 & \cdots & 0\\
			0 & \frac{1}{\theta_2} \bfgamma_{A_2} & \cdots & 0\\
			\vdots & \vdots & \ddots & \vdots \\
			0 & 0 & \cdots & \frac{1}{\theta_m} \bfgamma_{A_m}
		\end{pmatrix}
		\geq \bfgamma_{A_1, \dots, A_m},
	\end{align*}
	where $\bfgamma_{A_1, \dots, A_m}$ is the covariance matrix of $\rho_{A_1, \dots, A_m}$ and $\bfgamma_{A_j}$ is the covariance matrix of its marginal $\rho_{A_j}$.

\section{Discussion}\label{sec:discussion}

The maximal correlation is of particular interest for the local state transformation problem, where two parties are restricted to local operations but do not have access to classical communication. 
We have shown that the maximal correlation for bipartite Gaussian states can be simply calculated by restricting the optimization to local operators that are linear in terms of the phase-space quadrature operators. These optimal local operators may not be hermitian and therefore the maximal correlation, in general, is an upper bound on the classical maximal correlation of the joint outcome probability distribution associated with optimal local measurements (local hermitian observables). Using our results, one can investigate the problem of local state transformation with local operations when either the resource or the target state is Gaussian.

We have also introduced the Gaussian maximal correlation, as another measure of correlation for Gaussian states, by restricting the optimization to the class of hermitian and linear operators in terms of quadrature operators. This measure corresponds to performing optimal homodyne measurements on the phase-space quadratures, and is relevant to the local state transformation when both the resource and the target states are Gaussian. Nevertheless, we have shown through examples, that sometimes the maximal correlation gives stronger bounds on the local state transformation 
compared to the Gaussian maximal correlation even if both the states are Gaussian.

An interesting question, motivated by the definition of the Gaussian maximal correlation, is whether one can define other variants of maximal correlation by considering hermitian local operators that are quadratic or higher order in terms of quadrature operators. As we have shown for two-mode squeezed vacuum states, the optimal local measurements yielding the maximal correlation of one are photon-counting measurements that are non-Gaussian. This suggests that the maximization in the Gaussian maximal correlation can be further optimized using hermitian and quadratic operators in terms of quadrature operators. We leave the study of the properties of such invariants for future works.

We have also generalized the maximal correlation to the quantum maximal correlation ribbon for the multipartite case. We have shown that the quantum maximal correlation ribbon of Gaussian states can also be computed by restricting to local operators that are linear in terms of quadrature operators. Further, we have discussed its Gaussian version based on using Gaussian local observables.

In this paper, both in the bipartite and multipartite cases, in the computation of maximal correlation for Gaussian states we assume that each subsystem consists of a single mode. Another interesting problem it to generalize our results to cases where each party can have more than one mode of the shared Gaussian state. For example, one can consider the maximal correlation of a Gaussian state $\rho_{AB}$ in which each subsystem $A, B$ consists of two modes.

\emph{Hypercontractivity ribbon} is another invariant of quantum states that gives bounds on the local state transformation problem~\cite{DelgoshaBeigi2014}. It would be interesting to compute the hypercontractivity ribbon for Gaussian states. 

Finally, the quantum maximal correlation is really a measure of correlation and not a measure of entanglement, as it can take its maximum value on separable states. It is desirable to find a measure of entanglement that satisfies the tensorization property; see~\cite{beigi14} for an attempt  in this direction.


%

\newpage
\onecolumngrid
\appendix

\section{Quantum Maximal Correlation Ribbon}\label{app:MCRibbon}

In this appendix we discuss in details some properties of the quantum maximal correlation ribbon mentioned in Section~\ref{sec:MC-ribbon-Gaussian}. We start with the conditional expectation operator defined by
$$\cE_j(X) = \tr_{\neg j} (X\rho_{A^m} \rho_{A_j}^{-1}),$$
where $\tr_{\neg j}$ means partial trace with respect to all subsystems except the $j$-th one. We note that $\cE_j$ maps an operator acting on the whole system $A^m$ to an operator acting only on the $j$-th subsystem.

\begin{lemma}\label{lem:cE}
\begin{enumerate}[label=\emph{(\roman*)}]
\item If $Y$ is an operator acting on the $j$-th subsystem, then for any $X$ we have
$$\langle Y, X\rangle_{\rho_{A^m}} = \langle Y, \cE_j(X)\rangle_{\rho_{A_j}},$$
where the inner product is given in~\eqref{eq:def-inner-prod}.
\item $\cE_j$ is a projection, i.e., $\cE_j^2=\cE_j$.
\item $\cE_j$ is self-adjoint with respect to the inner product $\langle \cdot, \cdot\rangle_{\rho_{A^m}}$. 
\item For an operator $X$ define its \emph{variance} by
\begin{align}
\Var(X) &= \|X-\tr(\rho_{A^m}X) I\|_{\rho_{A^m}}^2\nonumber\\
&= \langle X-\tr(\rho_{A^m}X) I, X-\tr(\rho_{A^m}X) I\rangle_{\rho_{A^m}}\nonumber\\ 
&= \langle X, X\rangle_{\rho_{A^m}} - |\tr(\rho_{A^m}X)|^2. \label{eq:def-variance}
\end{align}
Then, we have
\begin{align}\label{eq:total-var}
\Var(X) = \Var(\cE_j(X)) + \Var(X-\cE_j(X)).
\end{align}
Moreover, since $\Var(X-\cE_j(X))\geq 0$ we have  $\Var(X)\geq \Var(\cE_j(X)).$
\end{enumerate}
\end{lemma}

Equation~\eqref{eq:total-var} can be understood as a quantum generalization of the \emph{law of total variance}. 

\begin{proof}
(i) Since $Y$ acts on the $j$-th subsystem we have
\begin{align*}
\langle Y, X\rangle_{\rho_{A^m}} &= \tr(\rho_{A^m} Y^\dagger X) = \tr(Y^\dagger X \rho_{A^m}) = \tr\big(Y^\dagger \tr_{\neg j}(X \rho_{A^m}) \big) =  \tr\big( \rho_{A_j}Y^\dagger \tr_{\neg j}(X \rho_{A^m}) \rho_{A_j}^{-1} \big)\\
&= \tr\big( \rho_{A_j}Y^\dagger \tr_{\neg j}(X \rho_{A^m} \rho_{A_j}^{-1})  \big)= \langle Y, \cE_j(X)\rangle_{\rho_{A_j}}.
\end{align*}

\noindent
(ii) We need to show that for $Y_{A_j}$, acting on the $j$-th subsystem, we have $\cE_j(Y)=Y$. We compute
$$\cE_j(Y)= \tr_{\neg j}(Y\rho_{A^m} \rho_{A_j}^{-1}) = Y \tr_{\neg j}(\rho_{A^m}) \rho_{A_j}^{-1}=Y \rho_{A_j} \rho_{A_j}^{-1}=Y.$$

\noindent
(iii) We need to show that $\langle \cE_j(X), Y\rangle_{\rho_{A^m}} = \langle X, \cE_j(Y)\rangle_{\rho_{A^m}}.$ We compute
\begin{align*}
\langle \cE_j(X), Y\rangle_{\rho_{A^m}} & = \tr(\rho_{A^m} \cE_j(X)^\dagger Y) = \tr\big( \rho_{A^m} \tr_{\neg j}(X\rho_{A^m} \rho_{A_j}^{-1})^\dagger Y  \big) \\
&= \tr\big(  \tr_{\neg j}(X\rho_{A^m} \rho_{A_j}^{-1})^\dagger \tr_{\neg j} (Y\rho_{A^m})  \big) = \tr\big(  \tr_{\neg j}( \rho_{A_j}^{-1} \rho_{A^m} X^\dagger) \tr_{\neg j} (Y\rho_{A^m})  \big) \\
&= \tr\big(  \rho_{A_j}^{-1} \tr_{\neg j}(  \rho_{A^m} X^\dagger) \tr_{\neg j} (Y\rho_{A^m})  \big) = \tr\big(   \tr_{\neg j}(  \rho_{A^m} X^\dagger) \tr_{\neg j} (Y\rho_{A^m}) \rho_{A_j}^{-1} \big) \\
&= \tr\big(   \tr_{\neg j}(  \rho_{A^m} X^\dagger) \tr_{\neg j} (Y\rho_{A^m} \rho_{A_j}^{-1})  \big)
= \tr\big(    \rho_{A^m} X^\dagger \tr_{\neg j} (Y\rho_{A^m} \rho_{A_j}^{-1})  \big) \\
&= \tr(\rho_{A^m} X^\dagger \cE_j(Y))\\
&= \langle X, \cE_j(Y)\rangle_{\rho_{A^m}}.
\end{align*}

\noindent
(iv) Let $Y=X-\cE_j(X)$. By (i) and the fact that $I_{A^m}$ can be considered as an operator acting on the $j$-th subsystem, we have $ \langle I_{A^m}, X\rangle_{\rho_{A^m}} = \langle I_{A_j}, \cE_j(X)\rangle_{\rho_{A_j}}$
 and $\tr(\rho_{A^m}Y) = \langle I_{A^m}, Y\rangle_{\rho_{A^m}} = 0$. This, in particular, means that $ \tr(\rho_{A^m}X)=  \tr(\rho_{A_j}\cE_j(X))$. Therefore,
\begin{align*}
\Var(Y) &= \langle Y, Y\rangle_{\rho_{A^m}}= \langle X-\cE_j(X), X-\cE_j(X)\rangle_{\rho_{A^m}} \\
&= \langle X, X\rangle_{\rho_{A^m}} -\langle X, \cE_j(X)\rangle_{\rho_{A^m}} - \langle \cE_j(X), X\rangle_{\rho_{A^m}} +\langle \cE_j(X), \cE_j(X)\rangle_{\rho_{A_j}}\\
&= \langle X, X\rangle_{\rho_{A^m}} -\langle \cE_j(X), \cE_j(X)\rangle_{\rho_{A_j}}\\
&= \langle X, X\rangle_{\rho_{A^m}} + |\tr(\rho_{A^m}X)|^2 -  |\tr(\rho_{A_j}\cE_j(X))|^2 -\langle \cE_j(X), \cE_j(X)\rangle_{\rho_{A_j}}\\
&= \Var(X) - \Var(\cE_j(X)),
\end{align*}
where in the third line we use (i).

\end{proof}

Now recall that the maximal correlation ribbon is defined by
\begin{align}\label{eq:app-def-MCR}
\fS(A_1, \dots, A_m) = \Big\{(\theta_1, \dots, \theta_m):\, \Var(X)\geq \sum_j \theta_j \Var(\cE_j(X)), \forall X\Big\}.
\end{align}
In the above definition (when the underlying Hilbert space is infinite dimensional), $X$ runs over the space $\cB_{A^m}$ of operators for which $\langle X, X\rangle_{\rho_A^m}$ is finite. We note that for such an $X\in \cB_{A^m}$, by the Cauchy-Schwarz inequality $\tr(\rho_{A^m} X)=\langle I,  X\rangle_{\rho_{A^m}}$ is also finite. Moreover, by the low of total variance established in Lemma~\ref{lem:cE}, $\cE_j$ maps $\cB_{A^m}$ to $\cB_{j}$, the space of operators $X_j$ acting on the $j$-th subsystem with $\langle X_j, X_j\rangle_{\rho_j}<\infty$. Thus, $\fS(A_1, \dots, A_m) $ is well-defined even in the infinite dimensional case.

To establish $(\theta_1, \dots, \theta_m)\in \fS(A_1, \dots, A_m) $ we need to verify an inequality for all operators $X\in \cB_{A^m}$ acting on the whole system $A^m$. In the following we show that we may restrict $X$ to be a linear combination of local operators belonging to $\cB_{j}$. We note that 
$$\cL=\overline{\cB_{1}+\cdots + \cB_{m}}\subset \cB_{A^m},$$ 
is a subspace of $\cB_{A^m}$, where $\overline{\cB_{1}+\cdots + \cB_{m}}$ is the closure of the subspace $\cB_{1}+\cdots + \cB_{m}$. Moreover, $\cB_{A^m}$ is equipped with the inner product $\langle \cdot, \cdot\rangle_{\rho_{A^m}}$. Thus, we may consider the orthogonal complement of $\cL$ in $\cB_{A^m}$:
$$\cK= \cL^{\perp}, \qquad \cB_{A^m} = \cL\oplus \cK=(\overline{\cB_{1}+\cdots + \cB_{m}})\oplus \cK.$$
Let $Y\in \cK$. Then, by the orthogonality condition and part (i) of Lemma~\ref{lem:cE} for $X_{A_j}\in \cB_{j}\subset \cL$  we have
$$0=\langle Y, X_{A_j}\rangle_{\rho_{A^n}}= \langle \cE_j(Y), X_{A_j}\rangle_{\rho_{A_j}}.$$
Thus, $\cE_j(Y)$ is an operator acting on the $j$-th subsystem that is orthogonal to all local operators in $\cB_{j}$. This means that $\cE_j(Y)=0$ for all $j$ and $Y\in \cK$. Moreover, since $I_{A^m}\in \cL$ we have $\tr(\rho_{A^m}Y)=\langle I_{A^m}, Y\rangle_{\rho_{A^m}} =0$.

\begin{proposition}\label{prop:L}
In the definition of the MC~ribbon in~\eqref{eq:app-def-MCR} we may restrict to $X\in \cL$. That is, we may restrict to operators that are linear combinations of local ones. Moreover, we may assume that $\tr(\rho_{A^m}X)=0$. 
\end{proposition}

\begin{proof}
First, by continuity we may assume that $X$ belongs to $(\cB_{1}+\cdots + \cB_{m})\oplus \cK\subseteq \overline{(\cB_{1}+\cdots + \cB_{m}})\oplus \cK = \cB_{A^m}$. 
Such an $X$ can be written as 
$$X = X_1+\cdots + X_m + Y,$$
where $X_j\in \cB_j$ and $Y\in \cK$. Let $X'=X_1+\cdots + X_m\in \cL$. Then, by the above discussion we have $\cE_j(X) = \cE_j(X')$. Moreover, using $\langle X', Y\rangle_{\rho_{A^m}}=0$ and $\tr(\rho_{A^m}Y)$ we have
\begin{align*}
\Var(X) &= \Var(X' + Y) =\langle X'+Y - \tr(\rho_{A^m}X'), X'+Y - \tr(\rho_{A^m}X')\rangle_{\rho_{A^m}}\\
&=\langle X' - \tr(\rho_{A^m}X'), X' - \tr(\rho_{A^m}X')\rangle_{\rho_{A^m}} + \langle Y , Y \rangle_{\rho_{A^m}}\\
& = \Var(X') + \Var(Y)\\
&\geq \Var(X').
\end{align*}
Putting these together we find that if the inequality in~\eqref{eq:app-def-MCR} holds for $X'$, then it holds for $X$. 

We also note that $\Var(X) = \Var(X+ cI_{A^n})$ and $\Var(\cE_j(X)) = \Var(\cE_j(X) + cI_{A_j})=\Var(\cE_j(X+ cI_{A^n})) $. Thus, without loss of generality we can assume that $X$ in~\eqref{eq:app-def-MCR} satisfies $\tr(\rho_{A^m}X)=0$.
\end{proof}

For any $1\leq j\leq m$ let 
\begin{align}\label{eq:def-B0}
\cB_j^0 = \{X\in \cB_j:\, \tr(\rho_{A_j} X)=0\},
\end{align}
and $\cL^0 = \overline{\cB_1^0+\cdots+ \cB_m^0}$. Observe that $\cB_j^0$ is the orthogonal complement of the identity operator in $\cB_j$, and we have
$$\cB_{A^m} = \mathbb CI\oplus \cL_m^0\oplus \cK.$$
By the above proposition, in the definition of the MC~ribbon we may restrict to $X\in \cB_1^0+\cdots+ \cB_m^0$.


\begin{theorem}\label{thm:MCR-equiv}
For any $1\leq j\leq m$ let 
$\{F_{j,k}: k=1, 2, \dots\}\subset \cB_j^0,$
be an orthonormal basis for $\cB_j^0$, where $\cB_j^0$ is defined in~\eqref{eq:def-B0}. Define
$$g_{jk, j'k'}:= \langle F_{j, k}, F_{j', k'}\rangle_{\rho_{A^m}},$$
and let $\mathcal G=(g_{jk, j'k'})$ be the Gram matrix associated with the set $\bigcup_j \{F_{j,k}: k=1, 2, \dots\}$. Also for $(\theta_1, \dots, \theta_m)\in [0,1]^m$ let $\Theta$ be the diagonal matrix 
\begin{align}\label{eq:def-Theta}
\Theta_{jk, j'k'}=\delta_{j, j'}\delta_{k, k'}\theta_j.
\end{align}
Then $(\theta_1, \dots, \theta_m)\in \fS(A_1, \dots, A_m)$ if and only if 
$\Theta^{-1}\geq \cG,$
meaning that $\Theta^{-1}- \cG$ is positive semidefinite. Equivalently, $(\theta_1, \dots, \theta_m)\in \fS(A_1, \dots, A_m)$ if and only if for every $X_j\in \cB_j^0$, $j=1, \dots, m$, we have
\begin{align}\label{eq:MCR-Var-equiv}
\Var\big(\sum_{j=1}^m X_j\big) \leq \sum_{j=1}^m \frac{1}{\theta_j} \Var(X_j).
\end{align}
\end{theorem}

\begin{proof}
By assumption any $X\in \cL^0$ can be written as
$$X= \sum_{j=1}^m \sum_k c_{jk} F_{j, k}.$$
Since $\tr(\rho_{A^m}X)=0$ we have
$$\Var(X) =\langle Y, Y\rangle_{\rho_{A^m}} = \sum_{j, j'=1}^m \sum_{k, k'} \bar c_{jk} c_{j'k'}  \langle F_{j, k}, F_{j', k'}\rangle_{\rho_{A^m}} = \bfc^\dagger \cG \bfc,$$
where $\bfc$ is the vector of coefficients $c_{jk}$. Next, using $\cE_j(X)\in \cB_j^0$ we have
\begin{align*}
\Var(\cE_j(X)) &= \langle \cE_j(X), \cE_j(X)\rangle_{\rho_{A^m}} \\
&= \|\cE_j(X)\|^2\\
&= \sum_k |\langle F_{j, k}, \cE_j(X)\rangle_{\rho_{A_j}} |^2 \\
&= \sum_k |\langle F_{j, k}, X\rangle_{\rho_{A^m}} |^2 \\
&= \sum_k \Big|\sum_{j', k'} c_{j'k'} g_{jk, j'k'} \Big|^2. 
\end{align*}
Also, a simple computation shows that 
$$\sum_{j=1}^m \theta_j \Var(\cE_j(X))  =\sum_{j=1}^m \theta_j \Big(\sum_k \Big|\sum_{j', k'} c_{j'k'} g_{jk, j'k'} \Big|^2\Big) = \bfc^{\dagger} G^\dagger \Theta \cG\bfc.$$
Therefore, $(\theta_1, \dots, \theta_m)\in \fS(A_1,\dots, A_m)$ if and only if 
$$\bfc^\dagger \cG \bfc \geq \bfc^{\dagger} \cG \Theta \cG\bfc.$$
Equivalently, this means that $\cG\geq  \cG^\dagger \Theta \cG$. Now note that $\cG$ is the Gram matrix of a linearly independent set, and is invertible. Hence, multiplying both sides from left and right by $\cG^{-1}=(\cG^\dagger)^{-1}$, we find that $\cG\geq  \cG \Theta \cG$ is equivalent to $\cG^{-1}\geq \Theta$. Next, using the fact that $t\mapsto 1/t$ is operator monotone~\cite{Bhatia-P}, we obtain the equivalence of $(\theta_1, \dots, \theta_m)\in \fS(A_1, \dots, A_m)$ and
$$\Theta^{-1}\geq \cG.$$
This equation means that for any $\bfc$ we have $\bfc^\dagger \Theta^{-1}\bfc \geq \bfc^\dagger \cG\bfc$. Then, thinking of $\bfc$ as the vector of coefficients of the expansions of operators $X_j\in\cB_j^0$ in  bases $\{F_{j, k}:\, k=1, 2, \dots\}$ for $j=1, \dots, m$, and following similar calculations as above, it is not hard to verify that the above equation is equivalent to~\eqref{eq:MCR-Var-equiv}.

\end{proof}


Let us use this theorem to compute the MC~ribbon for the example of a product state $\rho_{A^m}=\rho_{A_1}\otimes \cdots \rho_{A_m}$. In this case, for basis operators $F_{j, k}, F_{j', k'}$ with $j\neq j'$ we have
$$\langle F_{j, k}, F_{j', k'}\rangle_{\rho_{A^m}} = \tr(\rho_{A_j}\otimes \rho_{A_{j'}}  F_{j, k}^{\dagger}\otimes F_{j', k'}  ) 
= \tr(\rho_{A_j} F_{j, k}^{\dagger})\tr(\rho_{A_{j'}} F_{j', k'}  )=0.$$
Thus, the Gram matrix $\cG$ is the identity matrix. In this case, for any $(\theta_1, \dots, \theta_m)\in [0,1]^m$ and its associated matrix $\Theta$ we have $\Theta^{-1}\geq \cG$. This means that $\fS(A_1, \dots, A_m) =[0,1]^m$. That is, for product states, the MC~ribbon is the largest possible set. 


The following theorem shows that MC~ribbon is really a generalization of the maximal correlation in the bipartite case.

\begin{theorem}\label{thm:MCR-bipartite}
In the case of $m=2$ we have
$$\fS(A_1, A_2) = \Big\{(\theta_1, \theta_2)\in [0,1]^2:\, \big(\frac{1}{\theta_1} -1 \big)\big(\frac{1}{\theta_2} -1 \big)\geq \mu^2(A_1, A_2))\Big\}.$$
\end{theorem}

\begin{proof}
By Theorem~\ref{thm:MCR-equiv}, $(\theta_1, \theta_2)\in [0,1]^2$ belongs to $\fS(A_1, A_2)$ iff for any $X_j\in \cB_j^0$, $j=1, 2$ we have
$$\Var(X_1+X_2)\leq \frac{1}{\theta_1}\Var(X_1) +\frac{1}{\theta_2}\Var(X_2).$$
We note that since $\tr(\rho_{A_1}X_1)=\tr(\rho_{A_2}X_2)=0$ we have 
$$\Var(X_1+X_2) = \langle X_1, X_1\rangle_{\rho_{A_1}}+\langle X_2, X_2\rangle_{\rho_{A_2}} + 2\text{Re}\langle X_2, X_1\rangle_{\rho_{A^2}}.$$
Then, $(\theta_1, \theta_2)\in \fS(A_1, A_2)$ iff for every $X_j\in \cB_j^0$, $j=1, 2$ we have
$$2\text{Re}\langle X_2, X_1\rangle_{\rho_{A^2}}\leq \big(\frac{1}{\theta_1}-1\big) \Var(X_1) +\big(\frac{1}{\theta_2}-1\big) \Var(X_2).$$
Scaling $X_1, X_2$ and replacing them with $c_1X_1, c_2X_2$ we obtain
$$2\text{Re} \big(c_1\bar c_2 \langle X_2, X_1\rangle_{\rho_{A^2}}\big)\leq  |c_1|^2 \big(\frac{1}{\theta_1}-1\big) \Var(X_1) + |c_2|^2\big(\frac{1}{\theta_2}-1\big) \Var(X_2), \qquad \forall c_1, c_2.$$
This is equivalent to 
\begin{align*}
\begin{pmatrix}
\big(\frac{1}{\theta_1}-1\big) \Var(X_1) & - \langle X_1, X_2\rangle_{\rho_{A^2}}\\
-\langle X_2, X_1\rangle_{\rho_{A^2}} & \big(\frac{1}{\theta_2}-1\big) \Var(X_2)
\end{pmatrix} \geq 0,
\end{align*}
being positive semidefinite, that itself is equivalent to 
$$|\langle X_1, X_2\rangle_{\rho_{A^2}}| \leq \sqrt{\big(\frac{1}{\theta_1}-1\big)\big( \frac{1}{\theta_2}-1\big) }\sqrt{\Var(X_1)\Var(X_2)}.$$
Therefore, $(\theta_1, \theta_2)\in \fS(A_1, A_2)$ iff
\begin{align*}
\max\frac{|\langle X_1, X_2\rangle_{\rho_{A_2}}|}{\sqrt{\Var(X_1)\Var(X_2)}}\leq   \sqrt{\big(\frac{1}{\theta_1}-1\big)\big( \frac{1}{\theta_2}-1\big) },
\end{align*}
where the maximum is taken over $X_1\in \cB_1^0$ and $X_2\in \cB_2^0$. We note that by definition the left hand side is equal to $\mu(A_1, A_2)$.

\end{proof}

Recall that for $\ket\psi_{A_1A_2} = \frac{1}{\sqrt 2} (\ket{00}+\ket{11})$ we have $\mu(A_1, A_2)=1$. Thus, by the above theorem, $(\theta_1, \theta_2)\in \fS(A_1, A_2)$ iff $(1/\theta_1-1)(1/\theta_2-1)\geq 1$ that is equivalent to $\theta_1+\theta_2\leq 1$.
Therefore, the MC~ribbon for $\ket{\psi}_{A_1A_2}$ is the smallest possible set.

For the next result we first need a lemma.

\begin{lemma}\label{lem:box-times}
Let $X, X', Y, Y'$ be positive semidefinite matrices with the block forms $X=(X_{ij}), X'=(X'_{ij}), Y=(Y_{ij}), Y'=(Y'_{ij})$.
Suppose that $X\geq Y$ and $X'\geq Y'$, meaning that $X-Y$ and $X'-Y'$ are positive semidefinite. Then,
$$X\boxtimes X'\geq Y\boxtimes Y',$$
where $X\boxtimes X' $ is a matrix whose $ij$-th block equals $X_{ij}\otimes X'_{ij}$, and $Y\boxtimes Y'$ is defined similarly.
\end{lemma}

\begin{proof}
Since $X-Y$ and $X'-Y'$ as well as $X', Y$ are positive semidefinite, $(X-Y)\otimes X'$ and $Y\otimes (X'-Y')$ are positive semidefinite. This means that $(X-Y)\otimes X'+Y\otimes (X'-Y')= X\otimes X'-Y\otimes Y'$ is positive semidefinite. Now, we note that $X\boxtimes X'-Y\boxtimes Y'$ is a principal submatrix of $X\otimes X'-Y\otimes Y'$, so is positive semidefinite. 
\end{proof}

We now prove the main properties of the MC~ribbon, namely the data processing inequality and the tensorization.

\begin{theorem}\label{thm:MCR-monotonicity-ternsorization}
\begin{itemize}
\item  \textbf{(Data processing)} Suppose that $\Phi_{A_j\to A'_j}$ is a quantum operator acting on the subsystem $A_i$. Then for $\rho_{A_1, \dots, A_m}$ and $\sigma_{A'_1, \dots, A'_m}$ we have
$$\fS(A_1, \dots, A_m) \subseteq \fS(A'_1, \dots, A'_m).$$

\item  \textbf{(Tensorization)} For any $\rho_{A_1, \dots, A_m}, \sigma_{A'_1, \dots, A'_m}$ we have 
$$\fS(A_1A'_1, \dots, A_mA'_m) = \fS(A_1, \dots, A_m)\cap \fS(A'_1, \dots, A'_m),$$
where the left hand side is the MC~ribbon of the state $\rho_{A_1, \dots, A_m}\otimes \sigma_{A'_1, \dots, A'_m}$ considered as an $m$-partite state. 
\end{itemize}
\end{theorem}

\begin{proof}
(Data processing) Any quantum operation is a combination of an isometry and a partial trace. Clearly, local isometries do not change the MC~ribbon. Thus, it suffices to prove data processing under partial traces. To this end, we need to show that for a state $\rho_{A_1B_1, \dots, A_mB_m}$ we have
$$\fS(A_1B_1, \dots, A_mB_m) \subseteq \fS(A_1, \dots, A_m).$$
This inclusion is immediate once we note that in the definition of the MC~ribbon in~\eqref{eq:app-def-MCR}, we may restrict $X_{A_1B_1, \dots, A_mB_m}$ to act non-trivially only on the subsystems $(A_1, \dots, A_m)$.

\medskip
\noindent (Tensorization)
Let $\{F_{j,k}: k=0, 1, \dots\}$ be orthonormal bases for $\cB_{A_j}$ with $F_{j, 0}=I_{A_j}$.  We note that $\{F_{j,k}: k=1, 2, \dots\}$ is an orthonormal basis for $\cB_{A_j}^0 = \{X\in \cB_{A_j}:\, \langle I_{A_j}, X\rangle_{\rho_{A_j}} =0\}$.  
Then,  the Gram matrix $\cG$ of the set $\bigcup_j \{F_{j,k}: k=1, 2, \dots\}$ can be decomposed into blocks $\cG=(\cG_{jj'})$ where the block $\cG_{jj'}$ consists of the inner products of elements of $\{F_{j,k}: k=1, 2, \dots\}$ and $\{F_{j',k}: k=1, 2, \dots\}$.
Similarly, letting $\{F'_{j,k'}: k'=0, 1, \dots\}$ be an orthonormal basis for $\cB_{A'_j}$ with $F_{j, 0}=I_{A'_j}$, we can consider the associated Gram matrix $\cG'=(\cG'_{jj'})$ as above. 

We note that the space of operators acting on $A_jA'_{j}$ is equal to $\cB_{A_jA'_j} = \cB_{A_j}\otimes \cB_{A'_j}$. Then, 
$\{F_{j,k}\otimes F'_{j, k'}: k, k'=0, 1, \dots \}$ is an orthonormal basis for $\cB_{A_jA'_j}$. Moreover, since $F_{j, 0}\otimes F'_{j, 0} = I_{A_jA'_j}$, the set $\{F_{j,k}\otimes F'_{j, k'}: (k, k')\neq (0,0) \}$ is an orthonormal basis for $\cB^0_{A_jA'_j}$.
Thus, following Theorem~\ref{thm:MCR-equiv}, we need to compute the Gram matrix of the set $\bigcup_j \{F_{j,k}\otimes F'_{j, k'}: (k, k')\neq (0,0) \}$. This set can be decomposed into $\bigcup_j \{F_{j,k}\otimes F'_{j, k'}: (k, k')\neq (0,0) \} = \mathcal F\cup \mathcal F'\cup  \widehat{ \mathcal{F}}$, where 
\begin{align*}
\mathcal F&=  \bigcup_j  \{F_{j,k}\otimes I_{A'_j}: k=1, 2,\dots \},\\
\mathcal F' &=  \bigcup_j \{I_{A_j}\otimes F'_{j, k'}: k'=1, 2,\dots \},\\
\widehat{\mathcal F}& = \bigcup_j\{F_{j,k}\otimes F'_{j, k'}: k, k'=1, 2,\dots \}.
\end{align*}
Observe that, first, these three sets are orthogonal to each other. Second, the Gram matrix of $\mathcal F$ equals $\cG$, the Gram matrix of $\bigcup_j  \{F_{j,k}: k=1, 2,\dots \}$, and similarly the Gram matrix of $\mathcal F'$ equals $\cG'$. Third, a straightforward computation shows that the Gram matrix of $\widehat{\mathcal F}$ is equal to $\cG\boxtimes \cG'$, defined in Lemma~\ref{lem:box-times}. Putting these together, we find that the Gram matrix of the union is given by
\begin{align*}
 \begin{pmatrix}
\cG & 0 & 0\\
0 & \cG' & 0\\
0 & 0& \cG \boxtimes \cG'
\end{pmatrix}.
\end{align*}
Now, suppose that $(\theta_1, \dots, \theta_m)\in \fS(A_1, \dots, A_m)\cup \fS(A'_1, \dots, A'_m)$. Then, by Theorem~\ref{thm:MCR-equiv} we have $\Theta^{-1}\geq \cG$ and $\Theta'^{-1}\geq \cG'$, where $\Theta'$ is defined similarly to $\Theta$ but probably with a different size. Thus, by Lemma~\ref{lem:box-times} we have
\begin{align}\label{eq:G-G'-Upsilon}
\begin{pmatrix}
\Theta^{-1} & 0 & 0\\
0 & \Theta'^{-1} & 0\\
0 & 0& \Theta^{-1} \boxtimes \cG'
\end{pmatrix}
\geq  \begin{pmatrix}
\cG & 0 & 0\\
0 & \cG' & 0\\
0 & 0& \cG \boxtimes \cG'
\end{pmatrix}.
\end{align}
Next, we note that $\Theta^{-1}$ is diagonal, so 
$$\Theta^{-1} \boxtimes \cG' = \Theta^{-1}\boxtimes  \text{diag}(\cG'_{11}, \dots, \cG'_{mm}),$$
where $ \text{diag}(\cG'_{11}, \dots, \cG'_{mm})$ is a the block-diagonal matrix with blocks $\cG'_{jj}$ on the diagonal. On other hand, recall that $\cG'_{jj}$ is the Gram matrix of the set $\{F'_{j, k'}:\, k'=1, 2, \dots\}$, that is orthonormal. Thus, $\cG'_{jj}$ is the identity matrix. Therefore, 
$$\Theta^{-1}\boxtimes \cG' = \Theta^{-1}\boxtimes \text{diag}(I^{(1)}, \dots, I^{(m)}),$$ 
where $I^{(j)}$ is the identity matrix of an appropriate size. Then,  using Theorem~\ref{thm:MCR-equiv}, equation~\eqref{eq:G-G'-Upsilon} implies that $(\theta_1, \dots, \theta_m) \in \fS(A_1A'_1, \dots, A_mA'_m)$ and
$$\fS(A_1, \dots, A_m)\cap \fS(A'_1, \dots, A'_m)\subseteq \fS(A_1A'_1, \dots, A_mA'_m).$$
Inclusion in the other direction is a consequence of the data processing property proven in the first part.

\end{proof}

\section{Maximal Correlation Ribbon for Gaussian States}\label{app:Gaussian}

In this appendix we compute the MC~ribbon for multipartite Gaussian states. We show that, similar to the bipartite case, in order to compute the MC~ribbon via~\eqref{eq:MCR-Var-equiv}, it suffices to consider only observables $X_1, \dots, X_m$ that are linear in the position and momentum operators. To prove this result, we follow similar ideas used in Section~\ref{sec:bipart-Gaussian}.

Let $\rho_{A^m}=\rho_{A_1,\dots,A_m}$ be an $m$-mode Gaussian state. We note that, as is clear from its definition, the MC~ribbon does not change under local unitaries. Thus, by Theorem~\ref{prop:standard-form} we may assume that $\bfd(\rho_{A^m})=0$ and 
\begin{align}\label{eq:cov-multipartite}
\bfgamma=\bfgamma(\rho_{A^m}) = \begin{pmatrix}
\lambda_1 I & \bfnu_{12} & \cdots & \bfnu_{1m}\\
\bfnu_{21} & \lambda_2 I & \cdots & \bfnu_{2m}\\
\vdots  & \vdots & \ddots  & \vdots\\
\bfnu_{m1} & \bfnu_{m2} & \cdots  & \lambda_m I
\end{pmatrix},
\end{align}
where $\bfgamma(\rho_{A_j}) = \lambda_j I$ and $\bfnu_{jj'} =\bfnu_{j'j}^\top$.

\begin{theorem}\label{thm:gaussian-multipratite}
Let $\rho_{A^m}$ be an $m$-mode Gaussian state with first moment $\bfd(\rho^m)=0$ and covariance matrix given by~\eqref{eq:cov-multipartite}.  Let 
$$\Lambda ={\rm diag}\Big(\sqrt{\frac{\lambda_1+1}{2}}, \sqrt{\frac{\lambda_1-1}{2}}, \dots, \sqrt{\frac{\lambda_m+1}{2}}, \sqrt{\frac{\lambda_m-1}{2}}\Big),$$ 
and $\Upsilon_m=\Upsilon\oplus\cdots \oplus \Upsilon,$
where $\Upsilon$ is given in~\eqref{eq:def-M}. Then, $(\theta_1, \dots, \theta_m)\in [0,1]^m$ belongs to $\fS(A_1, \dots, A_m)$ if and only if 
\begin{align}\label{eq:MCR-G-lambda}
\frac 12 \Lambda^{-1} \bar \Upsilon_m (\bfgamma + i\Omega_m) \Upsilon_m^{\top} \Lambda^{-1}\leq \Theta^{-1}\otimes I_2,
\end{align}
where $\Theta={\rm diag}(  \theta_1,  \dots, \theta_m)$ and $I_2$ is the $2\times 2$ identity matrix.
In particular, to compute the MC~ribbon of $\rho_{A^m}$, it suffices to consider $X_1, \dots, X_m$ in~\eqref{eq:MCR-Var-equiv} that are linear in terms of 
phase-space quadrature operators. 
\end{theorem}

\begin{proof}
We follow similar steps as in the proof of Theorem~\ref{thm:main-mc}. First, we note that the set $\{H_{k, \ell}^{(j)}:\, k, \ell\geq 0\}$ defined via~\eqref{eq:def-H-k-ell}, for $\lambda=\lambda_j$, forms an orthonormal basis for $\cB_j$, the space of operators acting on $A_j$. Moreover, we have $H_{0,0}^{(j)} = I_{A_j}$, so $\{H_{k, \ell}^{(j)}:\, (k, \ell)\neq 0\}$ is an orthonormal basis for $\cB_{A_j}^0$. Thus, to apply Theorem~\ref{thm:MCR-equiv} we need to compute the Gram matrix of the set $\bigcup_j \{H_{k, \ell}^{(j)}:\, (k, \ell)\neq 0\}$. To this end, we decompose this set in terms of the total degrees:
$$\bigcup_j \{H_{k, \ell}^{(j)}:\, (k, \ell)\neq 0\} = \bigcup_{t=1}^\infty \Big(\bigcup_j \{H_{k, \ell}^{(j)}:\, k+\ell=t\}\Big).$$
We note that by~\eqref{eq:inner-prod-sum-zero}, basis operators with different degrees are orthogonal to each other. Then, the associated Gram matrix takes the form 
$$\text{diag}\big(\cG^{(1)}, \cG^{(2)}, \dots),$$ 
where $\cG^{(t)}$ is the Gram matrix of $\bigcup_j \{H_{k, \ell}^{(j)}:\, k+\ell=t\}$. As computed in~\eqref{eq:cQ-1}, the inner products of elements of $\{H_{k, \ell}^{(j)}:\, k+\ell=1\}$ and $\{H_{k, \ell}^{(j')}:\, k+\ell=1\}$ equals
$$
\cG^{(1)}_{jj'} = \frac 12 \begin{pmatrix}
\zeta_0^{-1}(\lambda_j) & 0\\
0 & \zeta_1^{-1}(\lambda_j)
\end{pmatrix}  \bar \Upsilon \bfnu_{jj'} \Upsilon^\top \begin{pmatrix}
\zeta_0^{-1}(\lambda_{j'}) & 0\\
0 & \zeta_1^{-1}(\lambda_{j'})
\end{pmatrix},
$$
where $\zeta_{0}(\lambda), \zeta_1(\lambda)$ are defined in~\eqref{eq:def-zeta}. On the other hand, a simple computation shows that 
$$\frac 12 \begin{pmatrix}
\zeta_0^{-1}(\lambda_j) & 0\\
0 & \zeta_1^{-1}(\lambda_j)
\end{pmatrix}  \bar \Upsilon (\lambda_j I + i\Omega) \Upsilon^\top \begin{pmatrix}
\zeta_0^{-1}(\lambda_{j}) & 0\\
0 & \zeta_1^{-1}(\lambda_{j})
\end{pmatrix}=I_2,
$$
is the identity matrix, i.e., the Gram matrix of $\{H_{k, \ell}^{(j)}:\, k+\ell=1\}$. Putting these together, we conclude that  
$$\cG^{(1)}=\frac 12 \Lambda^{-1} \bar \Upsilon_m (\bfgamma + i\Omega_m) \Upsilon_m^{\top} \Lambda^{-1},$$
which is the left hand side of~\eqref{eq:MCR-G-lambda}. Thus, the statement of the theorem says that $(\theta_1, \dots, \theta_m)\in \fS(A_1, \dots, A_m)$ if and only if
$$\cG^{(1)}\leq \Theta^{-1}\otimes I_2.$$
We note that by Theorem~\ref{thm:MCR-equiv}, $(\theta_1, \dots, \theta_m)\in \fS(A_1, \dots, A_m)$ if and only if for any $t$ we have
$$\cG^{(t)}\leq \Theta^{-1}\otimes I_{t+1},$$
where $\Theta=\text{diag}(\theta_1, \dots, \theta_m)$ and $I_{t+1}$ is the $(t+1)\times (t+1)$ identity matrix. Thus, to prove the theorem we need to show that if the above inequality holds for $t=1$, then it holds for all $t$. 
We note that $t=1$ corresponds to degree-one basis operators, that are linear in terms of quadrature operators.



Let $S$ be the matrix of size $(t+1)\times 2^t$ used in the proof of Theorem~\ref{thm:main-mc} whose entries are given by~\eqref{eq:def-S-matrix}. Also, let 
$$S_m = S\oplus \cdots \oplus S = I_m\otimes S.$$ 
Based on the computations in the proof of Theorem~\ref{thm:main-mc} we have $SS^{\dagger}=I$. Moreover, letting $\cG^{(t)}_{jj'}$ be the $jj'$-th block of $\cG^{(t)}$ that consists of the inner products of elements of $ \{H_{k, \ell}^{(j)}:\, k+\ell=t\}$ and $ \{H_{k, \ell}^{(j')}:\, k+\ell=t\}$, we have $S\big(\cG^{(1)}_{jj'}\big)^{\otimes t}S^\dagger = \cG^{(t)}$. Therefore, using the notation of Lemma~\ref{lem:box-times} we have
\begin{align}\label{eq:S-m-cG-t}
S_m \big(\cG^{(1)}\big)^{\boxtimes t} S^\dagger_m = \cG^{(t)}.
\end{align}
Thus, to prove the theorem we need to show that if $\cG^{(1)}\leq \Theta^{-1}\otimes I_2$, then $S_m \big(\cG^{(1)}\big)^{\boxtimes t} S^\dagger_m \leq \Theta^{-1}\otimes I_{t+1}$. Starting from $\cG^{(1)}\leq \Theta^{-1}\otimes I_2$ and using Lemma~\ref{lem:box-times} we have
$$\big(\cG^{(1)}\big)^{\boxtimes t}\leq \big(\Theta^{-1}\otimes I_2\big)\boxtimes \big(\cG^{(1)}\big)^{\boxtimes (t-1)} = \Theta^{-1}\otimes I_{2}^{\otimes t},$$
where the equality follows from the fact that $\Theta^{-1}\otimes I_2$ is diagonal and the blocks on the diagonal of $\big(\cG^{(1)}\big)^{\boxtimes (t-1)}$ are equal to $I_2^{\otimes (t-1)}$. Next, conjugating both sides with $S_m=I_m\otimes S $ and using~\eqref{eq:S-m-cG-t} yield
$$\cG^{(t)}\leq S_m\big(\Theta^{-1}\otimes I_{2}^{\otimes t}\big)S_m^\dagger = (I_m\otimes S)\big(\Theta^{-1}\otimes I_{2}^{\otimes t}\big)(I_m \otimes S^\dagger) = \Theta^{-1}\otimes SS^{\dagger} =\Theta^{-1}\otimes I_{t+1},$$
which proves the theorem.

\end{proof}

\end{document}